\documentclass[a4paper]{article}
\usepackage[UKenglish]{babel}
\usepackage{amsmath,amsfonts,amssymb,amsthm}
\usepackage[UKenglish]{babel}

\usepackage{amsbsy}
\usepackage{tikz}
\usepackage{times}
\usepackage{bm,bbm}
\usepackage{a4wide}
\usepackage[sans]{dsfont}
\usepackage[mathscr]{euscript}
\usepackage{cite}

\usepackage{enumerate}
\usepackage{color}
\usepackage{centernot}

\theoremstyle{plain}
\newtheorem{theorem}{Theorem}
\newtheorem{proposition}[theorem]{Proposition}

\newtheorem{lemma}[theorem]{Lemma}
\theoremstyle{remark}
\newtheorem{remark}{Remark}
\newtheorem{example}{Example}
\theoremstyle{definition}
\newtheorem{definition}{Definition}

\newcommand{\Var}{\operatorname{Var}}

\newcommand{\rmd}{\mathrm{d}}
\newcommand{\rmi}{\mathrm{i}}
\newcommand{\rme}{\mathrm{e}}

\newcommand{\Rbb}{{\mathbb{R}}}
\newcommand{\Cbb}{{\mathbb{C}}}

\newcommand{\Tr}{\operatorname{Tr}}

\newcommand{\abs}[1]{\left|#1\right|}
\renewcommand{\vec}[1]{\bm{#1}}

\newcommand{\id}{{\mathbbm{1}}}
\newcommand{\openone}{\id}

 \newcommand{\Ccal}{\mathcal{C}}

 \newcommand{\Ncal}{\mathcal{N}}

\newcommand{\Bscr}{\mathscr{B}}
 \newcommand{\Cscr}{\mathscr{C}}

\newcommand{\Gscr}{\mathscr{G}}
\newcommand{\Hscr}{\mathscr{H}}

\newcommand{\Lscr}{\mathscr{L}}

\newcommand{\Sscr}{\mathscr{S}}
\newcommand{\Tscr}{\mathscr{T}}

\newcommand{\Mo}{\mathsf{M}} 
\newcommand{\Qo}{\mathsf{Q}} 
\newcommand{\Po}{\mathsf{P}} 

\newcommand{\lh}{\Lscr(\Hscr)} 

\newcommand{\bor}[1]{\mathscr{B}(#1)} 
\newcommand{\vb}{\epsilon}

\begin{document}

\title{Measurement uncertainty relations for position and momentum: Relative entropy formulation}

\author{A. Barchielli, M. Gregoratti, A. Toigo}
\author{Alberto Barchielli$^{1,2,3}$,
Matteo Gregoratti$^{1,2}$, Alessandro Toigo$^{1,3}$
\\
\\
$^1$ Politecnico di Milano, Dipartimento di Matematica, \\
Piazza Leonardo da Vinci 32, I-20133 Milano, Italy,\\
$^2$ Istituto Nazionale di Alta Matematica (INDAM-GNAMPA),\\
$^3$ Istituto Nazionale di Fisica Nucleare (INFN), Sezione di Milano}

\maketitle
\begin{abstract}
Heisenberg's uncertainty principle has recently led to general measurement uncertainty relations for quantum systems: incompatible observables can be measured jointly or in sequence only with some unavoidable approximation, which can be quantified in various ways. The relative entropy is the natural theoretical quantifier of the information loss when a `true' probability distribution is replaced by an approximating one. In this paper, we provide a lower bound for the amount of information that is lost by replacing the distributions of the sharp position and momentum observables, as they could be obtained with two separate experiments, by the marginals of any smeared joint measurement. The bound is obtained by introducing an entropic error function, and optimizing it over a suitable class of covariant approximate joint measurements. We fully exploit two cases of target observables: (1) $n$-dimensional position and momentum vectors; (2) two components of position and momentum along different directions. In (1), we connect the quantum bound to the dimension $n$; in (2), going from parallel to orthogonal directions, we show the transition from  highly incompatible observables to compatible ones. For simplicity, we develop the theory only for Gaussian states and measurements.

\end{abstract}

\section{Introduction}

Uncertainty relations for position and momentum \cite{Hei27} have always been deeply related to the foundations of Quantum Mechanics. For several decades, their axiomatization has been of `preparation' type: an inviolable lower bound for the widths of the position and momentum distributions, holding in any quantum state. Such kinds of uncertainty relations, which are now known as \emph{preparation uncertainty relations} (PURs)\marginpar{PUR}, have been later extended to arbitrary sets of $n\geq 2$ observables \cite{SimMD94,Hol11,Hol01,Hol12}. All PURs trace back to the celebrated Robertson's formulation \cite{Rob29} of Heisenberg's uncertainty principle: for any two observables, represented by self-adjoint operators $A$ and $B$, the product of the variances of $A$ and $B$ is bounded from below by the expectation value of their commutator; in formulae, $\Var_\rho(A)\Var_\rho(B) \geq \frac{1}{4}\abs{\Tr\{\rho[A,B]\}}^2$, where $\Var_\rho$ is the variance of an observable measured in any system state $\rho$. In the case of position $Q$ and momentum $P$, this inequality gives Heisenberg's relation $\Var_\rho(Q)\Var_\rho(P) \geq \frac{\hbar^2}{4}$. About 30 years after Heisenberg and Robertson's formulation, Hirschman attempted a first statement of position and momentum uncertainties in terms of informational quantities. This led him to a formulation of PURs based on Shannon entropy \cite{Hir57}; his bound was later refined \cite{Beck75,B-BM75}, and extended to discrete observables \cite{MaaU88}. Also other entropic quantities have been used \cite{GibI11}. We refer to \cite{ColesBTW17,WehW10} for an extensive review on entropic PURs.

However, Heisenberg's original intent \cite{Hei27} was more focused on the unavoidable disturbance that a measurement of position produces on a subsequent measurement of momentum \cite{Oza15a,Wer04,Oza03a,BHL07,BLW14a,BLW14b,Oza03b,Oza02}.
Trying to give a better understanding of his idea, more recently new formulations were introduced, based on a `measurement' interpretation of uncertainty, rather than giving bounds on the probability distributions of the target observables. Indeed, with the modern development of the quantum theory of measurement and the introduction of \emph{positive operator valued measures} and \emph{instruments} \cite{BGL97,BLPY16,DavQTOS,Hol01,BarG09,HeZi12}, it became possible to deal with approximate measurements of incompatible observables and to formulate \emph{measurement uncertainty relations} (MURs)\marginpar{MUR} for position and momentum, as well as for more general observables. The MURs quantify the degree of approximation (or inaccuracy and disturbance) made by replacing the original incompatible observables with a \emph{joint approximate measurement} of them. A very rich literature on this topic flourished in the last 20 years, and various kinds of MURs have been proposed, based on distances between probability distributions, noise quantifications, conditional entropy, etc. \cite{BusHOW,Oza03a,Oza03b,Oza02,Oza15a,ColesBTW17,BLW13,BLW14a,BLW14b,BHL07,Wer04,Wer16,CF15,HSTZ14,BHSS13}.

In this paper, we develop a new information-theoretical formulation of MURs for position and momentum, using the notion of the \emph{relative entropy} (or \emph{Kullback-Leibler divergence}) of two probabilities. The relative entropy $S(p\|q)$ is an informational quantity which is precisely tailored to quantify the amount of information that is lost by using an approximating probability $q$ in place of the target one $p$.
Although classical and quantum relative entropies have already been used in the evaluation of the performances of quantum measurements \cite{BarL04,BarL05,BarL06a,BarL06b,Mac07,BarL07,BarL08,BusDW16,BarG09,CF15,BusHOW,Mac07}, their first application to MURs is very recent \cite{BarGT16}.

In \cite{BarGT16}, only MURs for discrete observables were considered. The present work is a first attempt to extend that information-theoretical approach to the continuous setting.
This extension is not trivial and reveals peculiar problems, that are not present in the discrete case. However, the nice properties  of the relative entropy, such as its scale invariance, allow for a satisfactory formulation of the entropic MURs also for position and momentum.

We deal with position and momentum in two possible scenarios. Firstly, we consider the case of $n$-dimensional position and momentum, since it allows to treat either scalar particles, or vector ones, or even the case of multi-particle systems. This is the natural level of generality, and our treatment extends without difficulty to it. Then, we consider a couple made up of one position and one momentum component along two different directions of the $n$-space. In this case, we can see how our theory behaves when one moves with continuity from a highly incompatible case (parallel components) to a compatible case (orthogonal ones).

The continuous case needs much care when dealing with arbitrary quantum states and approximating observables. Indeed, it is difficult to evaluate or even bound the relative entropy if some assumption is not made on probability distributions. In order to overcome these technicalities and focus on the quantum content of MURs, in this paper we consider only the case of Gaussian preparation states and Gaussian measurement apparatuses \cite{SimMD94,Hol11,Hol12,BvL05,WPGCRSL12,KiuS13,HKS15}. Moreover, we identify the class of the approximate joint measurements with the class of the joint POVMs satisfying the same symmetry properties of their target position and momentum observables \cite{Hol01,BGL97}. We are supported in this assumption by the fact that, in the discrete case \cite{BarGT16}, simmetry covariant measurements turn out to be the best approximations without any hypothesis (see also \cite{BLW13,BLW14a,BLW14b,Wer04,Wer16} for a similar appearance of covariance within MURs for different uncertainty measures).

We now sketch the main results of the paper. In the vector case, we consider approximate joint measurements $\Mo$ of the position $\vec Q\equiv (Q_1, \ldots , Q_n)$ and the momentum $\vec P\equiv (P_1,\ldots, P_n)$.
We find the following entropic MUR (Theorem \ref{prop:cinc}, Remark \ref{rem:MURvect}): for every choice of two positive thresholds $\vb_1,\vb_2$, with $\vb_1\vb_2\geq \hbar^2/4$, there exists a Gaussian state $\rho$ with position variance matrix $A^\rho\geq \vb_1\openone$ and momentum variance matrix $B^\rho\geq \vb_2\openone$ such that
\begin{equation}\label{vc}
S(\Qo_\rho\|\Mo_{1,\rho}) + S(\Po_\rho\|\Mo_{2,\rho})\geq  n\left(\log \rme\right) \left\{\ln \left(1+\frac \hbar {2\sqrt{\vb_1\vb_2}}\right)-\frac{\hbar}{\hbar+2\sqrt{\vb_1\vb_2}}\right\}
\end{equation}
for all Gaussian approximate joint measurements $\Mo$ of $\vec Q$ and $\vec P$.
Here $\Qo_\rho$ and $\Po_\rho$ are the distributions of position and momentum in the state $\rho$, and $\Mo_{\rho}$ is the distribution of $\Mo$ in the state $\rho$, with marginals $\Mo_{1,\rho}$ and $\Mo_{2,\rho}$; the two marginals turn out to be noisy versions of $\Qo_\rho$ and $\Po_\rho$.
The lower bound is strictly positive and it grows linearly with the dimension $n$.
The thresholds $\vb_1$ and $\vb_2$ are peculiar of the continuous case and they have a classical explanation: the relative entropy $S(p\|q)\to+\infty$ if the variance of $p$ vanishes faster than the variance of $q$, so that, given $\Mo$, it is trivial to find a state $\rho$ enjoying \eqref{vc} if arbtrarily small variances are allowed. What is relevant in our result is that
the total loss of information $S(\Qo_\rho\|\Mo_{1,\rho}) + S(\Po_\rho\|\Mo_{2,\rho})$ exceeds the lower bound even if we forbid target distributions with small variances.

The MUR \eqref{vc} shows that there is no Gaussian joint measurement which can approximate arbitrarily well both $\vec Q$ and $\vec P$. The lower bound \eqref{vc} is a consequence of the incompatibility between $\vec Q$ and $\vec P$ and, indeed, it vanishes in the classical limit $\hbar\to0$.
Both the relative entropies and the lower bound in \eqref{vc} are scale invariant. Moreover, for fixed $\vb_1$ and $\vb_2$, we prove the existence and uniqueness of an optimal approximate joint measurement, and we fully characterize it.

In the scalar case, we consider approximate joint measurements $\Mo$ of the position $Q_{\vec u}=\vec u\cdot\vec Q$ along the direction $\vec u$ and  the momentum $P_{\vec v}=\vec v\cdot \vec P$ along the direction $\vec v$, where $\vec u\cdot \vec v =\cos \alpha$. We find two different entropic MURs.
The first entropic MUR in the scalar case is similar to the vector case (Theorem \ref{prop:uvincdeg}, Remark \ref{rem:MURscal}).
The second one is (Theorem \ref{prop:min_sigma}):
\begin{equation}\label{sc2}
S(\Qo_{{\vec u},\rho}\|\Mo_{1,\rho}) + S(\Po_{{\vec v},\rho}\|\Mo_{2,\rho})\geq c_\rho(\alpha),
\end{equation}
\[
c_\rho(\alpha)
=\left(\log \rme\right) \left\{\ln \left(1+\frac{\hbar|\cos\alpha|}{2\sqrt{\Var\left(\Qo_{\vec u,\rho}\right) \Var\left(\Po_{\vec v,\rho}\right)}}\right)-\frac{\hbar|\cos\alpha|}{\hbar|\cos\alpha|+2\sqrt{\Var\left(\Qo_{\vec u,\rho}\right) \Var\left(\Po_{\vec v,\rho}\right)}}\right\},
\]
for all Gaussian states $\rho$ and all Gaussian joint approximate measurements $\Mo$ of $Q_{\vec u}$ and $P_{\vec v}$. This lower bound holds for every Gaussian state $\rho$ without constraints on the position and momentum variances $\Var\left(\Qo_{\vec u,\rho}\right)$  and $\Var\left(\Po_{\vec v,\rho}\right)$, it is strictly positive unless $\vec u$ and $\vec v$ are orthogonal, but it is state dependent. Again, the relative entropies and the lower bound are scale invariant.

The paper is organized as follows. In Section \ref{sec:stat+obs}, we introduce our target position and momentum observables, we discuss their general properties and define some related quantities (spectral measures, mean vectors and variance matrices, PURs for second order quantum moments, Weyl operators, Gaussian states). Section \ref{sec:rel+diff_ent} is devoted to the definitions and main properties of the relative and differential (Shannon) entropies. Section \ref{sec:PUR} is a review on the entropic PURs in the continuous case \cite{Hir57,Beck75,B-BM75}, with a particular focus on their lack of scale invariance. This is a flaw due to the very definition of differential entropy, and one of the reasons that lead us to introduce relative entropy based MURs.  In Section \ref{jointmeas} we construct the covariant observables which will be used as approximate joint measurements of the position and momentum target observables. Finally, in Section \ref{sec:entMURs} the main results on MURs that we sketched above are presented in detail. Some conclusions are discussed in Section  \ref{sec:concl}.

\section{Target observables and states}\label{sec:stat+obs}

Let us start with the usual position and momentum operators, which satisfy the canonical commutation rules:
\begin{equation}\label{vecQP}
\vec Q\equiv (Q_1, \ldots , Q_n), \qquad \vec P\equiv (P_1,
\ldots, P_n), \qquad \big[ Q_i,\,  P_j\big]=\rmi \hbar \delta_{ij}.
\end{equation}
Each of the vector operators has $n$ components; it could be the case of a single particle in one or more dimensions ($n=1,2,3$), or several scalar or vector particles, or the quadratures of $n$ modes of the electromagnetic field. We assume the Hilbert space $\Hscr$\marginpar{$\Hscr$} to be irreducible for the algebra generated by the canonical operators $\vec Q$ and $\vec P$. An \emph{observable} of the quantum system $\Hscr$ is identified with a \emph{positive operator valued measure} (POVM\marginpar{POVM}); in the paper, we shall consider observables with outcomes in $\Rbb^k$ endowed with its Borel $\sigma$-algebra $\bor{\Rbb^k}$\marginpar{$\bor{\Rbb^k}$}.
The use of POVMs to represent observables in quantum theory is standard and the definition can be found in many textbooks \cite{DavQTOS,BGL97,BLPY16,HeiMZ16}; the alternative name ``non-orthogonal resolutions of the identity'' is also used \cite{Hol01,Hol11,Hol12}.
Following \cite{BGL97,BLPY16,Hol12,HSTZ14}, a \emph{sharp observable} is an observable represented by a \emph{projection valued measure} (pvm\marginpar{pvm}); it is standard to identify a sharp observable on the outcome space $\Rbb^k$ with the $k$ self-adjoint operators corresponding to it by spectral theorem. Two observables are \emph{jointly measurable} or \emph{compatible} if there exists a POVM having them as marginals.
Because of the non-vanishing commutators, each couple $Q_i$, $P_i$, as well as the vectors $\vec Q$, $\vec P$, are not jointly measurable.

We denote by $\Tscr(\Hscr)$ the trace class operators on $\Hscr$,\marginpar{$\Sscr$, $\Tscr(\Hscr)$} by $\Sscr\subset \Tscr(\Hscr)$ the subset of the statistical operators (or states, preparations), and by $\lh$\marginpar{$\lh$} the space of the linear bounded operators.

\subsection{Position and momentum}\label{sec:ref}

Our target observables will be either $n$-dimensional position and momentum ({\em vector case}) or position and momentum along two different directions of $\Rbb^n$ ({\em scalar case}).
The second case allows to give an example ranging with continuity from maximally incompatible observables to compatible ones.

\subsubsection{Vector observables} \label{sec:refv}

As target observables we take $\vec Q$ and $\vec P$ as in \eqref{vecQP} and we denote by $\Qo(A), \Po(B)$, $A,B\in \Bscr(\Rbb^n)$, their pvm's, that is
\begin{equation}\label{eq:specQP}
Q_i = \int_{\Rbb^n} x_i \Qo(\rmd\vec{x}), \qquad P_i = \int_{\Rbb^n} p_i \Po(\rmd\vec{p}).
\end{equation}
Then, the distributions in the state $\rho\in\Sscr$ of a sharp  position and  a sharp  momentum measurements (denoted by $\Qo_\rho$ and $\Po_\rho$) are absolutely continuous with respect to the Lebesgue measure; we denote by $f(\bullet|\rho)$ and $g(\bullet|\rho)$ their probability densities: $\forall A,B\in \bor{\Rbb^n}$,
\begin{equation}\label{Q()P()}
\Qo_\rho(A)=\Tr\left\{\rho \Qo(A)\right\}=\int_{A}f(\vec x|\rho)\rmd \vec x, \qquad \Po_\rho(B)=\Tr\left\{\rho \Po(B)\right\}=\int_{B}g(\vec p|\rho)\rmd \vec p.
\end{equation}
In the Dirac notation, if  $|\vec x\rangle$ and $|\vec p\rangle$ are the improper position and momentum eigenvectors, these densities take the expressions
$f(\vec x|\rho)= \langle \vec x|\rho|\vec x\rangle$ and
$g(\vec p|\rho)= \langle \vec p|\rho|\vec p\rangle$, respectively.
The mean vectors and the variance matrices of these distributions will be given in \eqref{xstateVAR} and \eqref{pstateVAR}.

\subsubsection{Scalar observables} \label{sec:refs}

As target observables we take the position along a given direction $\vec u$ and the momentum along another given direction $\vec v$:
\begin{equation}\label{QuPv}
Q_{\vec u}=\vec u\cdot\vec Q, \qquad P_{\vec v}=\vec v\cdot \vec P, \qquad \text{with}\quad \vec u,\vec v\in \Rbb^n, \quad \abs{\vec u}=\abs{\vec v}=1,\quad \vec u\cdot \vec v =\cos \alpha.
\end{equation}
In this case we have $[Q_{\vec u},P_{\vec{v}}]=\rmi\hbar\cos \alpha$, so that $Q_{\vec u}$ and $P_{\vec{v}}$ are not jointly measurable, unless the directions $\vec u$ and $\vec v$ are orthogonal.

Their pvm's are denoted by $\Qo_{\vec u}$ and $\Po_{\vec v}$, their distributions in a state $\rho$ by $\Qo_{\vec u,\rho}$ and $\Po_{\vec v,\rho}$, and their corresponding probability densities by $f_{\vec u}(\bullet|\rho)$ and $g_{\vec v}(\bullet|\rho)$: $\forall A,B\in \bor{\Rbb}$,
\[
\Qo_{\vec u,\rho}(A) = \Tr\{\Qo_{\vec u}(A)\rho\}=\int_A f_{\vec u}(x|\rho)\,\rmd x,\qquad
\Po_{\vec v,\rho}(B) = \Tr\{\Po_{\vec v}(A)\rho\}= \int_B g_{\vec v}(p|\rho)\,\rmd p.
\]
Of course, the densities in the scalar case are marginals of the densities in the vector case.
Means and variances will be given in \eqref{uvmom}.

\subsection{Quantum moments.}
Let $\Sscr_2$ be the set of states for which the second moments of position and momentum are finite:\marginpar{$\Sscr_2$}
\[
\Sscr_2:=\left\{\rho\in \Sscr: \int_{\Rbb^n}\abs{\vec x}^2f(\vec x|\rho)\rmd  \vec x<+\infty, \
\int_{\Rbb^n}\abs{\vec p}^2g(\vec p|\rho)\rmd  \vec p<+\infty\right\}.
\]
Then, the mean vector and the variance matrix of the position $\vec Q$ in the state $\rho\in \Sscr_2$ are
\begin{equation}\label{xstateVAR}\begin{split}
a_i^\rho:=\int_{\Rbb^n}x_i f(\vec x|\rho)\rmd  \vec x&\equiv \Tr\left\{\rho Q_i\right\}, \\
\qquad A_{ij}^\rho:=\int_{\Rbb^n}\left(x_i-a_i^\rho\right) \left(x_j-a_j^\rho\right) f(\vec x|\rho)\rmd  \vec x&\equiv\Tr\left\{\rho \left(Q_i-a_i^\rho\right)\left(Q_j-a_j^\rho\right)\right\},
\end{split}\end{equation}
while for the momentum $\vec P$ we have
\begin{equation}\label{pstateVAR}\begin{split}
b_i^\rho :=\int_{\Rbb^n}p_i g(\vec p|\rho)\rmd  \vec p&\equiv \Tr\left\{\rho P_i\right\}, \\
\qquad B_{ij}^\rho:=\int_{\Rbb^n}\left(p_i-b_i^\rho\right) \left(p_j-b_j^\rho\right) g(\vec p|\rho)\rmd  \vec p&\equiv\Tr\left\{\rho \left(P_i-b_i^\rho\right)\left(P_j-b_j^\rho\right)\right\}.
\end{split}\end{equation}

For $\rho \in \Sscr_2$ it is possible to introduce also the mixed `quantum covariances'
\begin{equation}\label{xpstateVAR}
C_{ij}^\rho:=\Tr\left\{\rho\,\frac{(Q_i-a_i^\rho)(P_j-b_j^\rho)+(P_j-b_j^\rho)(Q_i-a_i^\rho)}{2}\right\}.
\end{equation}
Since there is no joint measurement for the position $\vec Q$ and momentum $\vec P$, the quantum covariances $C_{ij}^\rho$ are not covariances of a joint distribution, and thus they do not have a classical probabilistic interpretation.

By means of the moments above, we construct the three real $n\times n$ matrices $A^\rho,\, B^\rho, \, C^\rho$, the $2n$-dimensional vector $\mu^\rho$ and the symmetric $2n\times 2 n$ matrix $V^\rho$, with
\begin{equation}\label{muVV}
\mu^\rho :=\begin{pmatrix} \vec a^\rho \\ \vec b^\rho\end{pmatrix},\qquad V^\rho:= \begin{pmatrix}A^\rho & C^\rho  \\ (C^\rho)^T & B^\rho\end{pmatrix}.
\end{equation}
We say $V^\rho$ is the \emph{quantum variance matrix} of position and momentum in the state $\rho$. In \cite{SimMD94} dimensionless canonical operators are considered, but apart from this, our matrix $V^\rho$ corresponds to their ``noise matrix in real form''; the name ``variance matrix'' is also used \cite{KiuS13,SimSM87}.

In a similar way, we can introduce all the moments related to the position $Q_{\vec u}$ and momentum $P_{\vec v}$ introduced in \eqref{QuPv}. For $\rho\in \Sscr_2$, the means and variances are respectively
\begin{equation}\label{uvmom}
\vec u\cdot\vec a^\rho,\qquad \Var(\Qo_{\vec u,\rho})=\vec u \cdot A^\rho \vec u,
\qquad \qquad
\vec v\cdot\vec b^\rho,\qquad \Var(\Po_{\vec v,\rho})=\vec v \cdot B^\rho \vec v.
\end{equation}
Similarly to \eqref{xpstateVAR}, we have also the `quantum covariance' $\vec u \cdot C^\rho \vec v\equiv \vec v \cdot (C^\rho)^T \vec u$. Then, we collect the two means in a single vector and we introduce the variance matrix:
\begin{equation}\label{murho+V}
\mu^\rho_{\vec u,\vec v}:=\begin{pmatrix} \vec u\cdot\vec a^\rho\\ \vec v\cdot\vec b^\rho\end{pmatrix}, \qquad V^\rho_{\vec u,\vec v}:= \begin{pmatrix} \vec u \cdot A^\rho \vec u &
\vec u \cdot C^\rho \vec v \\ \vec u \cdot C^\rho \vec v & \vec v \cdot B^\rho \vec v\end{pmatrix}.
\end{equation}

\begin{proposition}
Let $V= \begin{pmatrix} A & C  \\ C^T & B\end{pmatrix}$ be a real symmetric $2n\times 2 n$ block matrix with the same dimensions of a quantum variance matrix. Define
\begin{equation}\label{def:Omega}
V_\pm :=\begin{pmatrix} A & C\pm \rmi \frac \hbar 2\,\openone \\ C^T\mp \rmi \frac \hbar 2\,
\openone & B\end{pmatrix} \equiv V \pm \frac{\rmi }2\,\Omega, \qquad \text{with} \qquad \Omega:= \begin{pmatrix} 0 & \hbar\openone \\ -\hbar\openone & 0\end{pmatrix}.
\end{equation}
Then
\begin{equation}\label{uncertmatrix}
\text{$V=V^\rho$ for some state $\rho\in \Sscr_2$} \iff V_+\geq 0 \iff V_-\geq 0.
\end{equation}
In this case we have: $V\geq 0$, $A>0$, $B>0$, and
\begin{equation}\label{RobUncert0}
(\vec u'\cdot A \vec u') (\vec v'\cdot  B \vec v')\geq \left(\vec v' \cdot C \vec u'\right)^2
+\frac{\hbar^2}4 \left(\vec v' \cdot\vec u'\right)^2,\qquad \forall \vec u'\in \Rbb^n, \quad \forall \vec v'\in \Rbb^n.
\end{equation}
\end{proposition}

The inequalities \eqref{uncertmatrix} for $V_\pm$ tell us exactly when a (positive semi-definite) real matrix $V$ is the quantum variance matrix of position and momentum in a state $\rho$. Moreover, they are the multidimensional version of the usual uncertainty principle expressed through the variances \cite{SimMD94,Hol01,Hol12}, hence they represent a form of PURs. The block matrix $\Omega$ in the definition of $V_\pm$ is useful to compress formulae involving position and momentum; moreover, it makes simpler to compare our equations with their frequent dimensionless versions (with $\hbar = 1$) in the literature \cite{KiuS13,HKS15}.

\begin{proof}
Equivalences \eqref{uncertmatrix} are well known, see e.g.\ \cite[Sect.\ 1.1.5]{Hol01}, \cite[Eq.~(2.20)]{Hol12}, \cite[Theor.~2]{SimMD94}. Then $V=\frac12 V_+ +\frac12 V_-\geq 0$.

By using the real block vector $\begin{pmatrix}\alpha\vec u'\\ \beta\vec v'\end{pmatrix}$, with arbitrary $\alpha,\beta\in \Rbb$ and given $\vec u',\, \vec v'\in \Rbb^n$, the semi-positivity \eqref{uncertmatrix} implies
\begin{equation*}
\begin{pmatrix}\vec u'\cdot A \vec u' & \vec u'\cdot C\vec v' \pm \rmi \frac \hbar 2 \vec u'\cdot\vec v' \\ \vec v' \cdot C^T\vec u' \mp \rmi \frac \hbar 2\vec v' \cdot\vec u' & \vec v'\cdot  B \vec v'\end{pmatrix}\geq 0, \qquad \forall \vec u'\in \Rbb^n, \quad \forall \vec v'\in \Rbb^n,
\end{equation*}
which in turn implies $A \geq 0$, $B \geq 0$ and \eqref{RobUncert0}.
Then, by choosing  $\vec u'=\vec v'=\vec u_i$, where $\vec u_1,\ldots,\vec u_n$ are the eigenvectors of $A$ (since $A$ is a real symmetric matrix, $\vec u_i\in\Rbb^n$ for all $i$), one gets the strict positivity of all the eigenvalues of $A$; analogously, one gets $B>0$.
\end{proof}

Inequality \eqref{RobUncert0} for $\vec u'=\vec u$ and $\vec v'=\vec v$ becomes the uncertainty rule \`a la Robertson \cite{Rob29} for the observables in \eqref{QuPv} (a position component and a momentum component spanning an arbitrary angle $\alpha$):
\begin{equation}\label{RobUncert}
\Var(\Qo_{\vec u,\rho})\, \Var(\Po_{\vec v,\rho})\geq \left(\vec v \cdot C^\rho \vec u\right)^2
+\frac{\hbar^2}4 \left(\cos\alpha\right)^2.
\end{equation}
Inequality \eqref{RobUncert} is equivalent to
\begin{equation}
V^\rho_{\vec u,\vec v} \pm \frac{\rmi \hbar}2\,\cos\alpha \begin{pmatrix} 0 & 1 \\ -1 & 0\end{pmatrix}\geq 0.
\end{equation}

Since $V_\pm$ are block matrices, their positive semi-definiteness can be studied by means of the Schur complements \cite{FZ05,Carlen10,Petz08}.
However, as $V_\pm$ are complex block matrices with a very peculiar structure, special results hold for them. Before summarizing the properties of $V_\pm$ in the next proposition, we need a simple auxiliary algebraic lemma.

\begin{lemma}\label{weylscor}
Let $A$ and $B$ be complex self-adjoint matrices such that $A\geq B\geq 0$. Then $\det A\geq \det B \geq 0$, and
the equality $\det A=\det B$ holds iff $A=B$.
\end{lemma}

\begin{proof}
Let $\lambda^\downarrow_i(A)$ and $\lambda^\downarrow_i(B)$ be the ordered decreasing sequences of the eigenvalues of $A$ and $B$, respectively. Then, by Weyl's inequality, $A\geq B\geq 0$ implies $\lambda^\downarrow_i(A)\geq\lambda^\downarrow_i(B)\geq 0$ for every $i$ \cite[Sect.\ III.2]{Bha97}. This gives the first statement. Moreover, if $A\geq B\geq 0$ and $\det A=\det B$, we get $\lambda^\downarrow_i(A)=\lambda^\downarrow_i(B)$ for every $i$. Then $A=B$ because $A-B\geq0$ and $\Tr\{A-B\}=0$.
\end{proof}

\begin{proposition}
Let $V= \begin{pmatrix}A & C  \\ C^T & B\end{pmatrix}$ be a real symmetric $2n\times 2 n$ matrix with the same dimensions of a quantum variance matrix. Then $V_+\geq 0$ (or, equivalently, $V_-\geq 0$) if and only if
$A>0$ and
\begin{equation}\label{ineqB>}
B\geq \left(C^T\mp \frac{\rmi\hbar}2\,\openone\right)A^{-1} \left(C\pm \frac{\rmi\hbar}2\,\openone\right)
\equiv C^T A^{-1}C +\frac{\hbar^2}4\, A^{-1}\mp\frac{\rmi\hbar}2\left(A^{-1}C-C^T A^{-1}\right).
\end{equation}
In this case we have
\begin{equation}\label{weakineqB>}
B\geq C^T A^{-1} C +\frac{\hbar^2}4\, A^{-1}\geq\frac{\hbar^2}4\, A^{-1}>0.
\end{equation}
Moreover, we have also the following properties for the various determinants:
\begin{equation}\label{uncdetdet}
(\det A)(\det B) \geq \det V= (\det A)\det\left(B-C^T A^{-1}C\right) \geq \left(\frac\hbar 2\right)^{2n},
\end{equation}
\begin{equation}\label{mus1}
\det V=\left(\frac \hbar 2 \right)^{2n} \quad \Leftrightarrow \quad B=C^T A^{-1} C +\frac {\hbar^2}4\, A^{-1} \quad \Rightarrow \quad  \quad C A = A C^T,
\end{equation}
\begin{equation}\label{mus2}
(\det A)(\det B)=\left(\frac \hbar 2 \right)^{2n} \quad \Leftrightarrow \quad B=\frac {\hbar^2}4\,A^{-1}, \quad C=0.
\end{equation}
\end{proposition}
By interchanging $A$ with $B$ and $C$ with $C^T$ in \eqref{ineqB>}-\eqref{mus2} equivalent results are obtained.
\begin{proof}
Since we already know that $V_+\geq 0$ implies the invertibility of $A$, the equivalence between \eqref{uncertmatrix} and  \eqref{ineqB>} with $A>0$  follows from \cite[Theor.~1.12 p.~34]{FZ05} (see also \cite[Theor.~11.6]{Petz08} or \cite[Lemma 3.2]{Carlen10}).

In \eqref{weakineqB>}, the first inequality follows by summing up the two inequalities in \eqref{ineqB>}. The last two ones are immediate by the positivity of $A^{-1}$.

The equality in \eqref{uncdetdet} is Schur's formula for the determinant of block matrices \cite[Theor.~1.1 p.~19]{FZ05}. Then, the first inequality is immediate by the lemma above and the trivial relation $B\geq B-C^T A^{-1} C$; the second one follows from \eqref{weakineqB>}:
\begin{equation*}
B-C^T A^{-1}C \geq\frac{\hbar^2}4\, A^{-1}
\quad \Rightarrow \quad \det\left(B-C^T A^{-1} C\right) \geq \det\left(\frac{\hbar^2}4\, A^{-1}\right)
=\frac{(\hbar/2)^{2n}}{\det A}.
\end{equation*}

The equality $\det V=\left(\frac \hbar 2 \right)^{2n}$ is equivalent to $\det\left(B-C^T A^{-1} C\right) = \det\left(\frac{\hbar^2}4\, A^{-1}\right)$; since the latter two determinants are evaluated on ordered positive matrices by \eqref{weakineqB>}, they coincide if and only if the respective arguments are equal (Lemma \ref{weylscor}); this shows the equivalence in \eqref{mus1}. Then, by \eqref{ineqB>}, the self-adjoint matrix $\frac{\rmi\hbar}2\left(A^{-1} C - C^T A^{-1}\right)$ is both positive semi-definite and negative semi-definite; hence it is null, that is, $C A = A C^T$.

Finally, $B=\frac {\hbar^2}4\,A^{-1}$ gives $(\det A)(\det B)=\left(\frac \hbar 2 \right)^{2n}$ trivially. Conversely, $(\det A)(\det B)=\left(\frac \hbar 2 \right)^{2n}$ implies $\det B=\det\Big(B-C^T A^{-1} C\Big)$ by \eqref{uncdetdet}; since $B\geq B-C^T A^{-1} C\geq0$ by \eqref{weakineqB>}, Lemma \ref{weylscor} then implies $C^T A^{-1}C=0$ and so $C=0$.
\end{proof}

By \eqref{ineqB>} and \eqref{weakineqB>}, every time three matrices $A,\, B, \, C$ define the quantum variance matrix of a state $\rho$, the same holds for $A,\, B, \, \widetilde C=0$. This fact can be used to characterize when two positive matrices $A$ and $B$ are the diagonal blocks of some quantum variance matrix, or two positive numbers $c_Q$ and $c_P$ are the position and momentum variances of a quantum state along the two directions $\vec{u}$ and $\vec{v}$.

\begin{proposition}\label{admissvar}
Two real matrices $A>0$ and $B>0$, having the dimension of the square of a length and momentum, respectively, are the diagonal blocks of a quantum variance matrix $V^\rho$ if and only if
\[
B\geq \frac {\hbar^2}4\,A^{-1}.
\]

Two real numbers $c_Q > 0$ and $c_P > 0$, having the dimension of the square of a length and momentum, respectively, are such that $c_Q = \Var(\Qo_{\vec u,\rho})$ and $c_P = \Var(\Po_{\vec v,\rho})$ for some state $\rho$ if and only if
\[
c_Q \, c_P \geq \left(\frac\hbar 2\,\cos\alpha\right)^2.
\]
\end{proposition}

\begin{proof}
For $A$ and $B$, the necessity follows from \eqref{weakineqB>}. \ The sufficiency comes from \eqref{ineqB>} by choosing \ $V^\rho= \begin{pmatrix}A & 0  \\ 0 & B\end{pmatrix}$.

For $c_Q$ and $c_P$, the necessity follows from \eqref{RobUncert0}. The sufficiency comes from \eqref{ineqB>} with $V^\rho= \begin{pmatrix}A & 0  \\ 0 & B\end{pmatrix}$ and for example the following choices of $A$ and $B$:
\begin{enumerate}[-]
\item if $\cos \alpha=\pm 1$, we take $A = c_Q\,\id$ and $B = c_P \, \id$;
\item if $\cos \alpha=0$, we let
\begin{align*}
A & = c_Q \, \vec{u}\vec{u}^T + \frac{\hbar^2}{4 c_P} \, \vec{v}\vec{v}^T + A' \qquad B = \frac{\hbar^2}{4 c_Q} \, \vec{u}\vec{u}^T + c_P \, \vec{v}\vec{v}^T + B' ,
\end{align*}
where $A'$ and $B'$ are any two scalar multiples of the orthogonal projection onto $\{\vec{u},\vec{v}\}^\perp$ satisfying $B'\geq \frac {\hbar^2}4\,A^{\prime\,-1}$ when restricted to $\{\vec{u},\vec{v}\}^\perp$;
\item if $\cos \alpha \notin\{0,\pm 1\}$, we choose
\begin{align*}
A & = c_Q \left[\vec{u}\vec{u}^T - \frac{1}{\cos\alpha}\, (\vec{u}\vec{v}^T + \vec{v}\vec{u}^T) + \frac{2}{(\cos\alpha)^2} \, \vec{v}\vec{v}^T\right] + A' \\
B & = \frac{c_P}{(\sin\alpha)^4} \left[\frac{(\sin\alpha)^2+(\cos\alpha)^4}{(\cos\alpha)^2} \, \vec{u}\vec{u}^T - \frac{1}{\cos\alpha}\, (\vec{u}\vec{v}^T + \vec{v}\vec{u}^T) + \vec{v}\vec{v}^T\right] + B' ,
\end{align*}
where $A'$ and $B'$ are as in the previous item.
\end{enumerate}
In the last two cases, we chose $A$ and $B$ in such a way that $B = \frac{c_Q \, c_P}{(\cos\alpha)^2} A^{-1}$ when restricted to the linear span of $\{\vec{u},\vec{v}\}$.
\end{proof}

\subsection{Weyl operators and Gaussian states}
In the following, we shall introduce Gaussian states, Gaussian observables and covariant observables on the phase-space. In all these instances, the Weyl operators are involved; here we recall their definition and some properties (see e.g.~\cite[Sect.\ 5.2]{Hol11} or \cite[Sect.\ 12.2]{Hol12}, where, however, the definition differs from ours in that the Weyl operators are composed with the map $\Omega^{-1}$ of \eqref{def:Omega}).

\begin{definition}
The \emph{Weyl operators} are the unitary operators defined by
\begin{equation} \label{Weyl}
W(\vec x,\vec p):= \exp\left\{\frac \rmi \hbar \left(\vec p\cdot \vec Q - \vec x\cdot \vec P\right)\right\}=\prod_{j=1}^n\rme^{\frac \rmi \hbar \left(p_jQ_j - x_jP_j\right)} = \prod_{j=1}^n\left(\rme^{\frac \rmi \hbar \,p_jQ_j}
\rme^{-\frac \rmi \hbar \,x_jP_j}\rme^{-\frac {\rmi  x_j p_j}{2\hbar}}\right).
\end{equation}
\end{definition}

The Weyl operators \eqref{Weyl} satisfy the composition rule
\[
W(\vec{x}_1,\vec{p}_1)W(\vec{x}_2,\vec{p}_2) = \exp\left\{-\frac{\rmi}{2\hbar}\left(\vec{x}_1\cdot\vec{p}_2 - \vec{x}_2\cdot\vec{p}_1\right)\right\} W(\vec{x}_1+\vec{x}_2,\vec{p}_1+\vec{p}_2);
\]
in particular, this implies the commutation relation
\begin{equation}\label{eq:comm}
W(\vec{x}_1,\vec{p}_1)W(\vec{x}_2,\vec{p}_2) = \exp\left\{-\rmi\begin{pmatrix}\vec{x}_1^T &\vec{p}_1^T \end{pmatrix}\Omega^{-1} \begin{pmatrix}\vec{x}_2 \\ \vec{p}_2 \end{pmatrix}\right\} W(\vec{x}_2,\vec{p}_2)W(\vec{x}_1,\vec{p}_1) .
\end{equation}
These commutation relations  imply the translation property
\begin{equation}\label{transl}
{W(\vec{x},\vec{p})}^*\, Q_i W(\vec{x},\vec{p}) = Q_i + x_i, \qquad {W(\vec{x},\vec{p})}^*\, P_i W(\vec{x},\vec{p}) = P_i + p_i , \qquad  i=1,\ldots,n ;
\end{equation}
due to this property, the Weyl operators are also known as \emph{displacement operators}.

With a slight abuse of notation, we shall sometimes use the identification
\begin{equation}\label{identification}
W(\vec x,\vec p)\equiv W\left(\begin{pmatrix} \vec{x} \\ \vec{p} \end{pmatrix}\right),
\end{equation}
where $\begin{pmatrix} \vec{x} \\ \vec{p} \end{pmatrix}$ is a block column vector belonging to the phase-space $\Rbb^n \times \Rbb^n\equiv \Rbb^{2n}$; here, the first block $\vec x$ is a position and the second block $\vec p$ is a momentum.

By means of the Weyl operators, it is possible to define the characteristic function of any trace-class operator.

\begin{definition}\label{def:charfunct} For any operator $\rho\in \Tscr(\Hscr)$, its \emph{characteristic function} is the complex valued function $\widehat \rho: \Rbb^{2n} \to \Cbb$ defined by
\begin{equation}\label{rho(w)}
\widehat \rho(w):= \Tr\left\{\rho W(-\Omega w)\right\},
\qquad w\equiv\begin{pmatrix}\vec k \\ \vec l \end{pmatrix}.
\end{equation}
\end{definition}
Note that $\vec k$ is the inverse of a length and $\vec l$ is the inverse of a momentum, so that $w$ is a block vector living in the space $\Rbb^{2n} \equiv \Rbb^n \times \Rbb^n$ regarded as the dual of the phase-space.

Instead of the characteristic function, sometimes the so called Weyl transform $\Tr \left\{ W(\vec x,\vec p)\rho\right\}$ is introduced \cite{KiuS13,Hol11}.

By \cite[Prop.\ 5.3.2, Theor.\ 5.3.3]{Hol11},
we have $\widehat \rho(w)\in L^2(\Rbb^{2n})$ and the following trace formula holds: $\forall \rho, \sigma \in \Tscr(\Hscr)$,
\begin{equation}\label{traceform}
\Tr\{\sigma^*\rho\}=\left(\frac \hbar{2\pi}\right)^n\int_{\Rbb^{2n}} \overline{\widehat \sigma(w)}\,\widehat \rho(w)\,\rmd  w.
\end{equation}
As a corollary
\cite[Coroll.\ 5.3.4]{Hol11}, we have that a state $\rho\in\Sscr$ is pure if and only if
\[
\left(\frac \hbar{2\pi}\right)^n\int_{\Rbb^{2n}} \abs{\widehat \rho(w)}^2\rmd w=1.
\]
By \cite[Lemma 3.1]{Wer83}, \cite[Prop.~8.5.(e)]{BLPY16}, the trace formula also implies
\begin{equation}\label{int=tr}
\frac{1}{(2\pi\hbar)^n} \int_{\Rbb^{2n}} W(\vec{x},\vec{p}) \rho {W(\vec{x},\vec{p})}^*\, \rmd\vec{x}\rmd\vec{p} = \Tr\{\rho\}\openone, \qquad \forall\rho\in\Tscr(\Hscr) \,.
\end{equation}
Moreover, the following inversion formula ensures that the characteristic function $\widehat \rho$ completely characterizes the state $\rho$ \cite[Coroll.\ 5.3.5]{Hol11}:
\[
\rho=\left(\frac \hbar{2\pi}\right)^n\int_{\Rbb^{2n}} W(\Omega w)\, \widehat \rho(w)\rmd  w, \qquad \forall\rho\in\Tscr(\Hscr) \,.
\]
The last two integrals are defined in the weak operator topology.

Finally, for $\rho\in \Sscr_2$, the moments \eqref{xstateVAR}--\eqref{muVV}
can be expressed as in \cite[Sect.\ 5.4]{Hol11}:
\begin{equation}\label{chmom}
-\rmi \, \frac{\partial \widehat \rho(w)}{\partial w_i}\Big|_{0}= \mu_i^\rho, \qquad
- \frac{\partial^2 \widehat \rho(w)}{\partial w_i\partial w_j}\Big|_{0}= V^\rho_{ij}+ \mu_i^\rho \mu_j^\rho.
\end{equation}

\begin{definition}[\cite{SimSM87,SimMD94,KiuS13,Hol01,Hol11,Hol12}]\label{def:Gstates}
A state $\rho$ is \emph{Gaussian} if
\begin{equation}\label{eq:char_gauss}\begin{split}
\widehat \rho(w)&=\exp\left\{\rmi  w^T \mu^\rho -\frac 1{2}\, w^TV^\rho w\right\}
\\ {} &= \exp\left\{\rmi  \left(\vec k \cdot \vec a^\rho + \vec l \cdot \vec b^\rho\right)-\frac 1{2}\left( \vec k \cdot A^\rho \vec k + \vec l \cdot B^\rho \vec l\right) - \vec k\cdot C^\rho \vec l \right\},
\end{split}\end{equation}
for a vector $\mu^\rho\in\Rbb^{2n}$ and a real $2n\times2n$ matrix $V^\rho
$ such that $V_+^\rho\geq0$.
\end{definition}

The condition $V_+^\rho\geq0$ is necessary and sufficient in order that the function \eqref{eq:char_gauss} defines the characteristic function of a quantum state \cite[Theor.~5.5.1]{Hol11}, \cite[Theor.~12.17]{Hol12}. Therefore, Gaussian states are exactly the states whose characteristic function is the exponential of a second order polynomial \cite[Eq.\ (5.5.49)]{Hol11}, \cite[Eq.\ (12.80)]{Hol12}.

We shall denote by  $\Gscr$\marginpar{$\Gscr$} the set of the Gaussian states; we have $\Gscr\subset \Sscr_2\subset \Sscr$. By \eqref{chmom}, the vectors $\vec a^\rho$, $\vec b^\rho$ and the matrices $A^\rho$, $B^\rho$, $C^\rho$ characterizing a Gaussian state $\rho$ are just its first and second order quantum moments introduced in \eqref{xstateVAR}--\eqref{xpstateVAR}.
By \eqref{eq:char_gauss}, the corresponding distributions of position and momentum are Gaussian, namely
\begin{equation}\label{Gsharpdistr}
\Qo_\rho=\Ncal(\vec a^\rho;A^\rho),\quad \Qo_{\vec u,\rho}=\Ncal(\vec u\cdot\vec a^\rho;\vec u\cdot A^\rho\vec u),\quad \Po_\rho=\Ncal(\vec b^\rho;B^\rho),\quad \Po_{\vec v,\rho}=\Ncal(\vec v\cdot\vec b^\rho;\vec v\cdot B^\rho\vec v).
\end{equation}

\begin{proposition}[Pure Gaussian states] For $\rho\in \Gscr$, we have  $\det V^\rho=\left(\frac \hbar 2 \right)^{2n}$ if and only if $\rho$ is pure.
\end{proposition}
\begin{proof}
The trace formula \eqref{traceform} and \eqref{eq:char_gauss} give $\Tr\{\rho^2\} = \frac{(\hbar/2)^n}{\sqrt{\det V^\rho}}$, and this implies the statement.
\end{proof}

\begin{proposition}[Minimum uncertainty states]\label{prop:minun} For $\rho\in \Sscr_2$, we have $(\det A^\rho)(\det B^\rho)=\left(\frac \hbar 2 \right)^{2n}$ if and only if $\rho$ is a pure Gaussian state and it factorizes into the product of minimum uncertainty states up to a rotation of $\Rbb^n$.
\end{proposition}
\begin{proof}
If $(\det A^\rho)(\det B^\rho)=\left(\frac \hbar 2 \right)^{2n}$, then  the equivalence \eqref{mus2} gives $B^\rho=\frac {\hbar^2}4\,(A^\rho)^{-1}$, so that the variance matrices $A^\rho$ and $B^\rho$ have a common eigenbasis $\vec u_1,\ldots,\vec u_n$. Thus, all the corresponding couples of position $\Qo_{\vec u_i}$ and momentum $\Po_{\vec u_i}$ have minimum uncertainties: $\Var(\Qo_{\vec u_i})\,\Var(\Po_{\vec u_i})=\frac{\hbar^2}{4}$. Therefore, if we consider the factorization of the Hilbert space $\Hscr=\Hscr_1\otimes\cdots\otimes\Hscr_n$ corresponding to the basis $\vec u_1,\ldots,\vec u_n$, all the partial traces of the state $\rho$ on each factor $\Hscr_i$ are minimum uncertainty states. Since for $n=1$ the minimum uncertainty states are pure and Gaussian, the state $\rho$ is a pure product Gaussian state.

The converse is immediate.
\end{proof}

\section{Relative and differential entropies} \label{sec:rel+diff_ent}

In this paper, we will be concerned with entropic quantities of classical type \cite{Top01,BuAn02,CovT06}.
We express them in `bits', that is we use the base-2 logarithms: $\log a \equiv \log_2 a$.

We deal only with probabilities on the measurable space $\big(\Rbb^n,\bor{\Rbb^n}\big)$ which admit densities with respect to the Lebesgue measure. So, we define the relative entropy and differential entropy only for such probabilities; moreover, we list only the general properties used in the following.

\subsection{Relative entropy or Kullback-Leibler divergence}\label{relentr}
The fundamental quantity is the \emph{relative entropy}, also called \emph{information divergence}, \emph{discrimination information}, \emph{Kullback-Leibler divergence} or \emph{information} or \emph{distance} or \emph{discrepancy}. The relative entropy of a probability $p$ with respect to a probability $q$ is defined for any couple of probabilities $p$, $q$ on the same probability space.

Given two probabilities $p$ and $q$ on $(\Rbb^n,\bor{\Rbb^n})$ with densities $f$ and $g$, respectively, the \emph{relative entropy} of $p$ with respect to $q$ is
\begin{equation}\label{def:relent}
S(p\|q)=\int_{\Rbb^n}f({\vec x})\log \frac{f({\vec x})}{g({\vec x})} \rmd  {\vec x}.
\end{equation}
The value $+\infty$ is allowed for $S(p\|q)$; the usual convention $0\log (0/0)=0$ is understood. The relative entropy \eqref{def:relent} is the amount of information that is lost when $q$ is used to approximate $p$ \cite[p.\ 51]{BuAn02}. Of course, if $\vec x$ is dimensioned, then the densities $f$ and $g$ have the same dimension (that is, the inverse of $\vec x$), and the argument of the logarithm is dimensionless, as it must be.

\begin{proposition}[\cite{CovT06}, Theorem 8.6.1] The following properties hold.
\begin{enumerate}[(i)]
\item $S(p\|q)\geq 0$.
\item $S(p\|q) = 0\qquad \iff \qquad p=q\qquad \iff \qquad f=g \quad\text{a.e.}$.
\item $S(p\|q)$ is invariant under a change of the unit of measurement.
\item If $p=\Ncal(\vec a;A)$ and $q=\Ncal(\vec b;B)$ with invertible variance matrices $A$ and $B$, then
\begin{equation}\label{Grelentr}
2\,S(p\|q)=(\log \rme)\Bigg\{\left(\vec a- \vec b\right)\cdot B^{-1}\left(\vec a- \vec b\right) +\Tr\left\{B^{-1}A-\openone\right\}\Bigg\}
+\log \frac{\det B}{\det A}\,.
\end{equation}
\end{enumerate}
\end{proposition}
As $S(p\|q)$ is scale invariant, it quantifies a relative error for the use of $q$ as an approximation of $p$, not an absolute one.

Let us employ the relative entropy to evaluate the effect of an additive Gaussian noise $\nu\sim\Ncal(b;\beta^2)$ on an independent Gaussian random variable $X$. If $X\sim\Ncal(a;\alpha^2)$, then $X+\nu\sim\Ncal(a+b;\alpha^2+\beta^2)$, and the relative entropy of the true distribution of $X$ with respect to its disturbed version $X+\nu$ is
$$S(X\|X+\nu)=\frac{\log \rme}{2}\,\frac{b^2-\beta^2}{\alpha^2+\beta^2} +\frac12\log \frac{\alpha^2+\beta^2}{\alpha^2}.$$
This expression vanishes if the noise becomes negligible with respect to the true distribution, that is if $\beta^2/\alpha^2\to0$ and $b^2/\alpha^2\to0$. On the other hand, $S(X\|X+\nu)$ diverges if the noise becomes too strong with respect to the true distribution, or, in other words, if the true distribution becomes too peaked with respect to the noise, that is, $\beta^2/\alpha^2\to+\infty$ or $b^2/\alpha^2\to+\infty$.

\subsection{Differential entropy}\label{sec:diffent}
The \emph{differential entropy} of an absolutely continuous random vector $\vec X$ with a probability density $f$ is
\[
H(\vec X):= -\int_{\Rbb^n} f({\vec x})\log f({\vec x}) \rmd  \vec  x .
\]
This quantity is commonly used in the literature, even if it lacks many of the nice properties of the Shannon entropy for discrete random variables. For example, $H(\vec X)$ is not scale invariant, and it can be negative \cite[p.\ 244]{CovT06}.

Since the density $f$ enters in the logarithm argument, the definition of $H(\vec X)$ is meaningful only when $f$ is dimensionless, which is the same as $\vec X$ being dimensionless. Note that, if $\vec X$ is dimensioned and $c >0$ is a real parameter making $\widetilde{\vec X}=c \vec X$ a dimensionless random variable, then
\[
H(\widetilde{\vec X})=-\int_{\Rbb^n} \frac{f(\vec u/c)}{c^n}\, \log \frac{f(\vec u/c)}{c^n} \,\rmd  \vec u
=-\int_{\Rbb^n} f(\vec x) \log \frac{f(\vec x)}{c^n} \,\rmd  \vec x \,.
\]
In the following, we shall consider the differential entropy only for dimensionless random vectors $\vec X$.
\begin{proposition}[\cite{CovT06}, Section 8.6]\label{prop:Hprops}
The following properties hold.
\begin{enumerate}[(i)]

\item If $\vec X$ is an absolutely continuous random vector with
variance matrix $A$, then
\[
H(\vec X)\leq \frac 12 \,\log \Big((2\pi\rme)^n \det A\Big)=
\frac n2\, \log \left(2\pi \rme\right)+\frac 12\, \Tr \log  A .
\]
The equality holds iff $\vec X$ is Gaussian with variance matrix $A$ and arbitrary mean vector $\vec a$.
\item If $\vec X=(X_1,\ldots,X_n)$ is an absolutely continuous random vector,
    then
\[
H(\vec X)\leq\sum_{i=1}^nH(X_i).
\]
The equality holds iff the components $X_1,\ldots,X_n$ are independent.
\end{enumerate}
\end{proposition}

\begin{remark}\label{rem:logdet}
In property (i) we have used the following well-known matrix identity, which follows by diagonalization:
\[
\log \det A=\Tr \log A, \qquad \forall A> 0.
\]
\end{remark}
\begin{remark}
Property (i) yields that the differential entropy of a Gaussian random variable $X\sim\Ncal(a;\alpha^2)$ is
$$
H(X)=\frac 12\, \log\left(2\pi \rme\alpha^2\right),
$$
which is an increasing function of the variance $\alpha^2$, and thus it is a measure of the uncertainty of $X$. Note that $H(X)\geq 0$ iff $\displaystyle\alpha^2\geq  1 /(2\pi\rme)$.
\end{remark}

\section{Entropic PURs for position and momentum}\label{sec:PUR}

The idea of having an entropic formulation of the PURs for position and momentum goes back to \cite{Hir57,Beck75,B-BM75}. However, we have just seen that, due to the presence of the logarithm, the Shannon differential entropy needs dimensionless probability densities. So, this leads us to introduce dimensionless versions of position and momentum.

Let $\lambda>0$ be a dimensionless parameter and $\varkappa$ a second parameter with the dimension of a mass times a frequency. Then, we introduce the dimensionless versions of position and momentum:
\begin{equation}\label{vecadim}
\widetilde{\vec Q}:=\sqrt{\frac \varkappa \hbar}\, \vec Q, \qquad \widetilde{ \vec P}=\frac {\lambda}{\sqrt
{\hbar \varkappa}}\,\vec P \qquad \Rightarrow \quad \left[\widetilde Q_i,\, \widetilde P_j\right]=\rmi \lambda\delta_{ij}.
\end{equation}
We use a unique dimensional constant $\varkappa$, in order to respect rotation symmetry and do not distinguish different particles.
Anyway, there is no natural link between the parameter multiplying $\vec Q$ and the parameter  multiplying $\vec P$; this is the reason for introducing $\lambda$.
As we see from the commutation rules, the constant $\lambda$ plays the role of a dimensionless version of $\hbar$; in the literature on PURs, often $\lambda=1$ is used \cite{B-BM75,Beck75,ColesBTW17}.

\subsection{Vector observables}\label{sec:Hvecobs}

Let $\widetilde\Qo$ and $\widetilde\Po$ be the pvm's of $\widetilde{\vec Q}$ and $\widetilde{\vec P}$; then, $\widetilde\Qo_\rho$ and $\widetilde\Po_\rho$ are their probability distributions in the state $\rho$. The total preparation uncertainty is quantified by the sum of the two differential entropies $H(\widetilde\Qo_\rho)+ H(\widetilde\Po_\rho)$. For $\rho\in \Gscr$, by Proposition \ref{prop:Hprops} we get
\begin{equation}\label{gaussdiffent}
H(\widetilde\Qo_\rho)+ H(\widetilde\Po_\rho)= n\log \left(\pi \rme\lambda\right)+
\frac 12 \,\log \left[\left(\frac{4}{\hbar^2}\right)^n\left(\det A^\rho\right)\left(\det B^\rho\right)\right].
\end{equation}
In the case of product states of minimum uncertainty, we have $\left(\det A^\rho\right)\left(\det B^\rho\right)=\left(\hbar^2/4\right)^n$; then, by taking \eqref{uncdetdet} into account, we get
\begin{equation}\label{entunc}
\inf_{\rho\in \Gscr}\Big\{H(\widetilde\Qo_\rho)+ H(\widetilde\Po_\rho)\Big\}= n\log \left(\pi \rme\lambda\right).
\end{equation}
Thus, the bound \eqref{entunc} arises from quantum relations between $\vec Q$ and $\vec P$; indeed, there would be no lower bound for \eqref{gaussdiffent} if we could take both $\det A^\rho$ and $\det B^\rho$ arbitrarily small.

By item (ii) of Proposition \ref{prop:Hprops}, the differential entropy for the distribution of a random vector is smaller than the sum of the entropies of its marginals; however, the final bound \eqref{entunc} is a tight bound for both $H(\widetilde\Qo_\rho)+ H(\widetilde\Po_\rho)$ and $\sum_{i=1}^nH(\widetilde\Qo_{i,\rho})+ \sum_{i=1}^nH(\widetilde\Po_{i,\rho})$.

By the results of \cite{B-BM75,Beck75}, the same bound \eqref{entunc} is obtained even if the minimization is done over all the states, not only the Gaussian ones.

The uncertainty result \eqref{entunc} depends on $\lambda$, this being a consequence of the lack of scale invariance of the differential entropy; note that the bound is positive if and only if $\lambda>1/(\pi \rme)$. Sometimes in the literature the parameter $\hbar$ appears in the argument of the logarithm \cite{BusHOW,CF15}; this fact has to be interpreted as the appearance of a parameter with the numerical value of $\hbar$, but without dimensions. In this sense the formulation \eqref{entunc} is consistent with both the cases with $\lambda=1$ or $\lambda=\hbar$. Sometimes the smaller bound $\ln 2\pi$ appears in place of $\log \pi \rme$ \cite{MaaU88}; this is connected to a state dependent formulation of the entropic PUR
\cite[Sect.\ V.B]{ColesBTW17}.

\subsection{Scalar observables}
The dimensionless versions of the scalar observables introduced in \eqref{QuPv} are
\begin{equation}\label{adim}
\widetilde Q_{\vec u}=\sqrt{\frac \varkappa \hbar}\,  Q_{\vec u}, \qquad \widetilde P_{\vec v}=\frac {\lambda}{\sqrt
{\hbar \varkappa}}\, P_{\vec v} \quad \Rightarrow \quad \left[\widetilde Q_{\vec u},\, \widetilde P_{\vec v}\right]=\rmi \lambda\cos\alpha.
\end{equation}
We denote by $\widetilde\Qo_{\vec u,\rho}$ and $\widetilde\Po_{\vec v,\rho}$ the associated distributions in the state $\rho$. For $\rho\in\Sscr_2$, the respective means and variances are
\[
\sqrt{\frac \varkappa \hbar}\,\vec u\cdot\vec a^\rho, \qquad \frac {\lambda}{\sqrt{\hbar \varkappa}}\,\vec v\cdot\vec b^\rho,  \qquad
\Var(\widetilde\Qo_{\vec u,\rho})=\frac \varkappa \hbar\,\vec u \cdot A^\rho \vec u, \qquad
\Var(\widetilde\Po_{\vec v,\rho})=\frac {\lambda^2}{\hbar \varkappa}\,\vec v \cdot B^\rho \vec v,
\]
with $\sqrt{\Var(\widetilde\Qo_{\vec u,\rho})\,\Var(\widetilde\Po_{\vec v,\rho})}\geq \lambda\abs{\cos\alpha}/2$.

As in the vector case, the total preparation uncertainty is quantified by the sum of the two differential entropies $H(\widetilde\Qo_{\vec u,\rho})+ H(\widetilde\Po_{\vec v,\rho})$. For $\rho\in \Gscr$, Proposition \ref{prop:Hprops} gives
\begin{equation}\label{scgaussdiffent}
H(\widetilde\Qo_{\vec u,\rho})+ H(\widetilde\Po_{\vec v,\rho})=\log\left(2\pi\rme \sqrt{\Var(\widetilde\Qo_{\vec u,\rho})\,\Var(\widetilde\Po_{\vec v,\rho})}\right).
\end{equation}
Then, we have the lower bound
\begin{equation}\label{PUR}
\inf_{\rho\in \Gscr}\Big\{H(\widetilde\Qo_{\vec u,\rho})+ H(\widetilde\Po_{\vec v,\rho})\Big\}=\log \left(\pi\rme\lambda\abs{\cos\alpha}\right)
=\frac{1+\ln \left(\pi\abs{\lambda\cos\alpha}\right)}{\ln 2},
\end{equation}
which depends on $\lambda$, but not on  $\varkappa$. Of course, because of \eqref{scgaussdiffent}, for Gaussian states a lower bound for the sum $H(\widetilde\Qo_{\vec u,\rho})+ H(\widetilde\Po_{\vec v,\rho})$ is equivalent to a lower bound for the product $\Var(\widetilde\Qo_{\vec u,\rho})\,\Var(\widetilde\Po_{\vec v,\rho})$. By a slight generalization of the results of \cite{B-BM75,Beck75}, the bound \eqref{PUR} is obtained also when the minimization is done over all the states.

Let us note that the bound in \eqref{PUR} is positive for $\abs{\lambda\cos\alpha}>1/(\pi\rme)$, and it goes to $-\infty$ for $\alpha\to \pi/2$, which is the case of compatible $\Qo_{\vec u,\rho}$ and $\Po_{\vec v,\rho}$.
In the case $\alpha=0$, the bound \eqref{PUR} is the same as \eqref{entunc} for $n=1$.

\section{Approximate joint measurements of position and momentum}\label{jointmeas}

In order to deal with MURs for position and momentum observables, we have to introduce the class of approximate joint measurements of position and momentum, whose marginals we will compare with the respective sharp observables. As done in \cite{Hol01,Hol11,BHL07,CHT04}, it is natural to characterize such a class by requiring suitable properties of covariance under the group of space translations and velocity boosts: namely, by {\em approximate joint measurement of position and momentum} we will mean any POVM on the product space of the position and momentum outcomes sharing the same covariance properties of the two target sharp observables. As we have already discussed, two approximation problems will be of our concern: the approximation of the position and momentum vectors (vector case, with outcomes in the  phase-space $\Rbb^n\times \Rbb^n$), and the approximation of one position and one momentum component along two arbitrary directions (scalar case, with oucomes in $\Rbb\times \Rbb$). In order to treat the two cases altogether, we consider POVMs with outcomes in $\Rbb^m \times \Rbb^m\equiv \Rbb^{2m} $, which we call \emph{bi-observables}; they correspond to a measurement of $m$ position components and $m$ momentum components. The specific covariance requirements will be given in the Definitions \ref{def:cov_ph-sp}, \ref{def:phi_cov}, \ref{def:Jcov}.

In studying the properties of probability measures on $\Rbb^k$, a very useful notion is that of the characteristic function, that is, the Fourier cotransform of the measure at hand; the analogous quantity for POVMs turns out to have the same relevance. Different names have been used in the literature to refer to the characteristic function of POVMs, or, more generally, quantum instruments, such as characteristic operator or operator characteristic function \cite{Hol78,BarL85,BarL91,BarHL93,BarG09,Hol86,Hol01,BarL04,KiuS13}. As a variant, also the symplectic Fourier transform quite often appears \cite[Sect. 12.4.3]{Hol12}. The characteristic function has been used, for instance, to study the quantum analogues of the infinite-divisible distributions \cite{BarL85,Hol86,BarL91,BarHL93,Hol01,BarL04} and measurements of Gaussian type \cite{Hol78,Hol12,KiuS13}. Here, we are interested only in the latter application, as our approximating bi-observables will typically be Gaussian. Since we deal with bi-observables, we limit our definition of the characteristic function only to POVMs on $\Rbb^m\times\Rbb^m$, which have the same number of variables of position and momentum type.

Being measures, POVMs can be used to construct integrals, whose theory is presented e.g.~in \cite[Sect.~4.8]{BLPY16}, \cite[Sect.\ 2.9, Prop.\ 2.9.1]{Hol11}.

\begin{definition}
Given a bi-observable $\Mo : \bor{\Rbb^{2m}}\to\lh$, the \emph{characteristic function} of $\Mo$ is the operator valued function $\widehat{\Mo} : \Rbb^{2m}\to\lh$, with
\begin{equation}\label{eq:def_char}
\widehat{\Mo}(\vec{k},\vec{l}) = \int_{\Rbb^{2m}} \rme^{\rmi(\vec{k}\cdot\vec{x} + \vec{l}\cdot\vec{p})} \Mo(\rmd\vec{x}\rmd\vec{p}) .
\end{equation}
\end{definition}
Here, the dimensions of the vector variables $\vec{k}$ and $\vec{l}$ are the inverses of a length and momentum, respectively, as in the definition of the characteristic function of a state \eqref{rho(w)}. This definition is given so that $\Tr\left\{\widehat{\Mo}(\vec{k},\vec{l})\rho\right\}$ is the usual characteristic function of the probability distribution $\Mo_\rho$ on $\Rbb^{2m}$.

\subsection{Covariant vector observables}

In terms of the pvm's \eqref{eq:specQP}, the translation property \eqref{transl} is equivalent to the symmetry properties
\[
W(\vec{x},\vec{p}) \Qo(A) {W(\vec{x},\vec{p})}^* = \Qo(A+\vec{x}), \qquad W(\vec{x},\vec{p}) \Po(B) {W(\vec{x},\vec{p})}^* = \Po(B+\vec{p}), \qquad \forall A,B\in\bor{\Rbb^n} ,
\]
and they are taken as the transformation property defining the following class of POVMs on $\Rbb^{2n}$ \cite{Wer83,BLPY16,BGL97,KiuS13,CHT04}.

\begin{definition}\label{def:cov_ph-sp}
A \emph{covariant phase-space observable} is a bi-observable $\Mo : \bor{\Rbb^{2n}} \to \lh$ satisfying the covariance relation
\begin{equation*}
W(\vec{x},\vec{p})\Mo(Z){W(\vec{x},\vec{p})}^* = \Mo\left(Z + \begin{pmatrix} \vec{x} \\ \vec{p} \end{pmatrix} \right), \qquad \forall Z\in\bor{\Rbb^{2n}},\quad\forall\vec{x},\vec{p}\in\Rbb^n .
\end{equation*}
We denote by $\Cscr$\marginpar{$\Cscr$} the set of all the covariant phase-space observables.
\end{definition}

The interpretation of covariant phase-space observables as approximate joint measurements of position and momentum is based on the fact that their marginal POVMs
\[
\Mo_1(A) = \Mo(A\times\Rbb^n), \qquad \Mo_2(B) = \Mo(\Rbb^n\times B), \qquad A,B\in\bor{\Rbb^n},
\]
have the same symmetry properties of $\Qo$ and $\Po$, respectively. Although $\Qo$ and $\Po$ are not jointly measurable, the following well-known result says that there are plenty of covariant phase-space observables \cite{CDT03,KLY06}, \cite[Theor.~4.8.3]{Hol11}. In \eqref{Mo} below, we use the parity operator $\Pi$ on $\Hscr$, which is such that
\begin{equation}\label{parity}
\Pi\, W(\vec x,\vec p)\,\Pi=W(-\vec x,-\vec p)={W(\vec x,\vec p)}^*.
\end{equation}

\begin{proposition}\label{prop:C=S}
The covariant phase-space observables are in one-to-one correspondence with the states on $\Hscr$, so that we have the identification $\Sscr\sim \Cscr$; such a correspondence  $\sigma \leftrightarrow \Mo^\sigma$ is given by
\begin{equation}\label{Mo}
\begin{aligned}
\Mo^\sigma(B) & =\int_B   M^\sigma(\vec x,\vec p)\,\rmd \vec x \rmd  \vec p , \qquad \forall B\in \bor{\Rbb^{2n}},
\\
M^\sigma(\vec x,\vec p) & =\frac{1}{\left(2\pi \hbar\right)^n}\,  W(\vec x,\vec p) \Pi \sigma \Pi {W(\vec x,\vec p)}^* .
\end{aligned}
\end{equation}

\end{proposition}

The characteristic function \eqref{eq:def_char} of a measurement $\Mo^\sigma\in\Cscr$ has a very simple structure in terms of the characteristic function \eqref{rho(w)} of the corresponding state $\sigma \in \Sscr$.

\begin{proposition}
The characteristic function of $\Mo^\sigma\in \Cscr$ is given by
\begin{equation}\label{hatMosigma}
\widehat \Mo^\sigma(\vec{k},\vec{l})  = W\left(-\Omega w\right) \widehat \sigma(w) , \qquad w\equiv\begin{pmatrix}\vec{k} \\ \vec{l}\end{pmatrix}\in \Rbb^{2n},
\end{equation}
and the characteristic function of the probability $\Mo_\rho^\sigma$ is
\begin{equation}\label{prod_ft}
\Tr\left\{\widehat \Mo^\sigma(\vec{k},\vec{l})\rho\right\}  = \widehat \rho(w) \widehat \sigma(w).
\end{equation}
\end{proposition}
In \eqref{hatMosigma} we have used the identification \eqref{identification}.
The characteristic function of a state is introduced in \eqref{rho(w)}.
\begin{proof} By the commutation relations \eqref{eq:comm}, we have
\[
W(-\hbar\vec{l},\hbar\vec{k})W(\vec{x},\vec{p}){W(-\hbar\vec{l},\hbar\vec{k})}^*= \rme^{i(\vec{k}\cdot\vec{x} + \vec{l}\cdot\vec{p})} W(\vec{x},\vec{p}).
\]
Then, we get
\begin{align*}
\widehat \Mo^\sigma(\vec{k},\vec{l}) & = \frac{1}{(2\pi\hbar)^n} \int_{\Rbb^{2n}} \rme^{i(\vec{k}\cdot\vec{x} + \vec{l}\cdot\vec{p})} W(\vec{x},\vec{p})\Pi\sigma\Pi {W(\vec{x},\vec{p})}^*\, \rmd\vec{x}\rmd\vec{p} \\
& = \frac{1}{(2\pi\hbar)^n} \int_{\Rbb^{2n}} W(-\hbar\vec{l},\hbar\vec{k})W(\vec{x},\vec{p}){W(-\hbar\vec{l},\hbar\vec{k})}^*\,\Pi\sigma\Pi {W(\vec{x},\vec{p})}^* \,\rmd\vec{x}\rmd\vec{p}\\
& = W(-\hbar\vec{l},\hbar\vec{k}) \Tr\{{W(-\hbar\vec{l},\hbar\vec{k})}^*\,\Pi\sigma\Pi\},
\end{align*}
where we have used formula \eqref{int=tr}. By \eqref{parity} and definition \eqref{rho(w)}, we get \eqref{hatMosigma}. Again by \eqref{rho(w)}, we get \eqref{prod_ft}.
\end{proof}

In terms of probability densities, measuring $\Mo^\sigma$ on the state $\rho$ yields the density function
$h^\sigma(\vec{x},\vec{p}|\rho) = \Tr\{M^\sigma(\vec{x},\vec{p})\rho\}$. Then, by \eqref{prod_ft}, the densities of the marginals $\Mo^\sigma_{1,\rho}$ and $\Mo^\sigma_{2\,\rho}$ are the convolutions
\begin{equation}\label{h*}
h^{\sigma}_1(\bullet|\rho) = f(\bullet|\rho) \ast f(\bullet|\sigma), \qquad h^{\sigma}_2(\bullet|\rho) = g(\bullet|\rho) \ast g(\bullet|\sigma),
\end{equation}
where $f$ and $g$ are the sharp densities introduced in \eqref{Q()P()}.
By the arbitrariness of the state $\rho$, the marginal POVMs of $\Mo^\sigma$ turn out to be the convolutions  (or `smearings')
\[
\Mo^\sigma_1 (A)
\int_A\rmd\vec{x}\int_{\Rbb^{n}}f(\vec{x}-\vec{x}'|\sigma) \Qo(\rmd\vec{x}'), \qquad
\Mo^\sigma_2 (B)
\int_B\rmd\vec{p}\int_{\Rbb^{n}}g(\vec{p}-\vec{p}'|\sigma) \Po(\rmd\vec{p}')
\]
(see e.g.~\cite[Sect.~III, Eqs.~(2.48), (2.49)]{BGL97}).

Let us remark that the distribution of the approximate position observable $\Mo^\sigma_1 $ in a state $\rho$ is the distribution of the sum of two independent random vectors: the first one is distributed as the sharp position $\Qo$ in the state $\rho$, the second one is distributed as the sharp position $\Qo$ in the state $\sigma$. In this sense, the approximate position $\Mo^\sigma_1 $ looks like a sharp position plus an independent noise given by $\sigma$. Of course, a similar fact holds for the momentum. However, this statement about the distributions can not be extended to a statement involving the observables. Indeed, since $\Qo$ and $\Po$ are incompatible, nobody can jointly observe $\Mo^\sigma$, $\Qo$ and $\Po$, so that the convolutions \eqref{h*} do not correspond to sums of random vectors that actually exist when measuring $\Mo^\sigma$.

\subsection{Covariant scalar observables}

Now we focus on the class of approximate joint measurements of the observables $Q_{\vec{u}}$ and  $P_{\vec{v}}$ representing position and momentum along two possibly different directions $\vec{u}$ and $\vec{v}$ (see  Section \ref{sec:refs}). As in the case of covariant phase-space observables, this class is defined in terms of the symmetries of its elements: we require them to transform as if they were joint measurements of $\Qo_{\vec{u}}$ and $\Po_{\vec{v}}$. Recall that $\Qo_{\vec{u}}$ and $\Po_{\vec{v}}$ denote the spectral measures of $Q_{\vec{u}}$, $P_{\vec{v}}$.

Due to the commutation relation \eqref{eq:comm}, the following covariance relations hold
\begin{align*}
W(\vec{x},\vec{p})\Qo_{\vec{u}}(A){W(\vec{x},\vec{p})}^* = \Qo_{\vec{u}}(A + \vec{u}\cdot\vec{x}), \qquad W(\vec{x},\vec{p})\Po_{\vec{v}}(B){W(\vec{x},\vec{p})}^* = \Po_{\vec{v}}(B + \vec{v}\cdot\vec{p}),
\end{align*}
for all $A,B\in\bor{\Rbb}$ and $\vec{x},\vec{p}\in\Rbb^n$. We employ covariance to define our class of approximate joint measurements of $\Qo_{\vec{u}}$ and $\Po_{\vec{v}}$.

\begin{definition}\label{def:phi_cov}
A \emph{$(\vec u,\vec v)$-covariant bi-observable} is a POVM $\Mo:\bor{\Rbb^2}\to\lh$ such that
\[
W(\vec{x},\vec{p})\Mo(Z){W(\vec{x},\vec{p})}^* = \Mo\left(Z + \begin{pmatrix} \vec{u}\cdot\vec{x} \\ \vec{v}\cdot\vec{p} \end{pmatrix}\right), \qquad \forall Z\in\bor{\Rbb^2},\quad \forall\vec{x}, \vec{p}\in\Rbb^n .
\]
We denote by $\Cscr_{\vec u, \vec v}$\marginpar{$\Cscr_{\vec u, \vec v}$} the class of such bi-observables.
\end{definition}

So, our approximate joint measurements of $\Qo_{\vec{u}}$ and $\Po_{\vec{v}}$ will be all the bi-observables in the class $\Cscr_{\vec u, \vec v}$.

\begin{example}
The marginal of a covariant phase-space observable $M^\sigma$ along the directions $\vec u$ and $\vec v$ is a $(\vec u,\vec v)$-covariant bi-observable. Actually, it can be proved that, if $\cos \alpha \neq 0$, all $(\vec u,\vec v)$-covariant bi-observables can be obtained in this way.
\end{example}

It is useful to work with a little more generality, and merge Definitions \ref{def:cov_ph-sp} and \ref{def:phi_cov} into a single notion of covariance.

\begin{definition}\label{def:Jcov}
Suppose $J$ is a $k\times 2n$ real matrix. A POVM $\Mo:\bor{\Rbb^k}\to\lh$ is a \emph{$J$-covariant observable} on $\Rbb^k$ if
\[
W(\vec{x},\vec{p})\Mo(Z){W(\vec{x},\vec{p})}^* = \Mo\left(Z + J \begin{pmatrix} \vec{x} \\ \vec{p}\end{pmatrix}\right), \qquad \forall Z\in\bor{\Rbb^k},\qquad \forall\vec{x},\vec{p}\in\Rbb^n .
\]
\end{definition}

Thus, approximate joint observables of $\Qo_{\vec{u}}$ and $\Po_{\vec{v}}$ are just $J$-covariant observables on $\Rbb^2$ for the choice of the $2\times 2n$ matrix
\begin{equation}\label{eq:def_phi}
J = \begin{pmatrix} \vec{u}^T & \vec{0}^T \\ \vec{0}^T & \vec{v}^T \end{pmatrix} \,.
\end{equation}
On the other hand, covariant phase-space observables constitute the class of $\id_{2n}$-covariant observables on $\Rbb^{2n}$, where $\id_{2n}$ is the identity map of $\Rbb^{2n}$.

\subsection{Gaussian measurements}\label{sec:gaussM}

When dealing with Gaussian states, the following class of bi-observables quite naturally arises.

\begin{definition}
A POVM $\Mo : \bor{\Rbb^{2m}}\to\lh$ is a {\em Gaussian bi-observable} if
\begin{equation}\label{eq:char_gauss_POVM}
\widehat{\Mo}(\vec{k},\vec{l}) = W\left(-\Omega (J^\Mo)^T \begin{pmatrix} \vec{k} \\ \vec{l} \end{pmatrix}\right) \exp\left\{\rmi \begin{pmatrix} \vec k^T & \vec l^T\end{pmatrix} \begin{pmatrix} \vec a^\Mo \\ \vec b^\Mo\end{pmatrix} - \frac{1}{2} \begin{pmatrix} \vec k^T & \vec l^T\end{pmatrix} V^\Mo \begin{pmatrix} \vec k \\ \vec l\end{pmatrix} \right\}
\end{equation}
for two vectors $\vec{a}^\Mo,\vec{b}^\Mo\in\Rbb^m$, a real $2m\times 2n$ matrix $J^\Mo$ and a real symmetric $2m\times 2m$ matrix $V^\Mo$ satisfying the condition
\begin{equation}\label{eq:def_pos_gauss_POVM}
V^\Mo  \pm\frac{\rmi}{2} J^\Mo \Omega (J^\Mo)^T\geq 0 .
\end{equation}
We set $\mu^\Mo=\begin{pmatrix}\vec{a}^\Mo\\ \vec{b}^\Mo\end{pmatrix}$. The triple $(\mu^\Mo,\,V^\Mo,\,J^\Mo)$ is the set of the \emph{parameters} of the Gaussian observable $\Mo$.
\end{definition}

In this definition, the vector $\vec{a}^\Mo$ has the dimension of a length, and $\vec{b}^\Mo$ of a momentum; similarly, the matrices $J^\Mo$, $V^\Mo$ decompose into blocks of different dimensions.
The condition \eqref{eq:def_pos_gauss_POVM} is necessary and sufficient in order that the function \eqref{eq:char_gauss_POVM} defines the characteristic function of a POVM.

For unbiased Gaussian measurements, i.e., Gaussian bi-observables with $\vec{a}^\Mo = \vec{b}^\Mo = \vec{0}$, the previous definition coincides with the one of \cite[Section 12.4.3]{Hol12}. It is also a particular case of the more general definition of Gaussian observables on arbitrary (not necessarily symplectic) linear spaces that is given in \cite{KiuS13,HKS15}. We refer to \cite{Hol12,KiuS13} for the proof that Eq.~\eqref{eq:char_gauss_POVM} is actually the characteristic function of a POVM.

Measuring the Gaussian observable $\Mo$ on the Gaussian state $\rho$ yields the probability distribution $\Mo_\rho$ whose characteristic function is
\begin{multline*}
\Tr\{\widehat{\Mo}(\vec{k},\vec{l})\rho\} = \widehat{\rho}\left((J^\Mo)^T\begin{pmatrix} \vec{k} \\ \vec{l} \end{pmatrix}\right) \exp\left\{\rmi \begin{pmatrix} \vec k^T & \vec l^T\end{pmatrix} \begin{pmatrix} \vec a^\Mo \\ \vec b^\Mo\end{pmatrix} - \frac{1}{2} \begin{pmatrix} \vec k^T & \vec l^T\end{pmatrix} V^\Mo \begin{pmatrix} \vec k \\ \vec l\end{pmatrix} \right\} \\
{}= \exp\left\{\rmi \begin{pmatrix} \vec k^T & \vec l^T\end{pmatrix} \left[\begin{pmatrix} \vec a^\Mo \\ \vec b^\Mo\end{pmatrix} + J^\Mo \begin{pmatrix} \vec a^\rho \\ \vec b^\rho\end{pmatrix} \right] - \frac{1}{2} \begin{pmatrix} \vec k^T & \vec l^T\end{pmatrix} \left[V^\Mo + J^\Mo V^\rho (J^\Mo)^T\right] \begin{pmatrix} \vec k \\ \vec l\end{pmatrix} \right\};
\end{multline*}
hence the output distribution is Gaussian,
\begin{equation}\label{eq:hMo}
\Mo_\rho = \Ncal \left( J^\Mo \mu^\rho+\mu^\Mo  ;\, J^\Mo V^\rho (J^\Mo)^T + V^\Mo \right) .
\end{equation}

\subsubsection{Covariant Gaussian observables}

For Gaussian bi-observables, $J$-covariance has a very easy characterization.

\begin{proposition}\label{prop:gauss_Tcov}
Suppose $\Mo$ is a Gaussian bi-observable on $\Rbb^{2m}$ with parameters $(\mu^\Mo,\, V^\Mo,\, J^\Mo)$. Let $J$ be any $2m\times 2n$ real matrix. Then, the POVM $\Mo$ is a $J$-covariant observable if and only if $J^\Mo = J$.
\end{proposition}
\begin{proof}
For $\vec{x},\,\vec{p}\in\Rbb^{n}$, we let $\Mo^\prime$ and $\Mo^{\prime\prime}$ be the two POVMs on $\Rbb^{2m}$ given by
\[
\Mo^\prime(Z) = W(\vec{x},\vec{p})\Mo(Z){W(\vec{x},\vec{p})}^*, \qquad \Mo^{\prime\prime}(Z) = \Mo\left(Z+J\begin{pmatrix} \vec{x} \\ \vec{p}\end{pmatrix}\right) , \qquad \forall Z\in\bor{\Rbb^{2m}} .
\]
By the commutation relations \eqref{eq:comm} for the Weyl operators, we immediately get
\begin{multline*}
\widehat  \Mo^\prime(\vec{k},\vec{l})  = W(\vec{x},\vec{p})\widehat{\Mo}(\vec{k},\vec{l}){W(\vec{x},\vec{p})}^* = \exp\left\{-\rmi \begin{pmatrix} \vec{x}^T & \vec{p}^T\end{pmatrix} \Omega^{-1} \left[-\Omega (J^\Mo)^T \begin{pmatrix} \vec{k} \\ \vec{l} \end{pmatrix}\right]\right\}\widehat{\Mo}(\vec{k},\vec{l})
\\
{} = \exp \left\{-\rmi\begin{pmatrix} \vec{k}^T & \vec{l}^T \end{pmatrix} J^\Mo \begin{pmatrix} \vec{x} \\ \vec{p} \end{pmatrix} \right\} \widehat{\Mo}(\vec{k},\vec{l}) ;
\end{multline*}
we have also
\begin{multline*}
\widehat  \Mo^{\prime\prime}(\vec{k},\vec{l})  = \int_{\Rbb^{2m}} \exp\left\{\rmi\begin{pmatrix} \vec{k}^T & \vec{l}^T \end{pmatrix} \left[\begin{pmatrix} \vec{x}' \\ \vec{p}' \end{pmatrix} - J\begin{pmatrix} \vec{x} \\ \vec{p} \end{pmatrix} \right]\right\} \Mo(\rmd\vec{x}'\rmd\vec{p}') \\
{}= \exp \left\{-\rmi\begin{pmatrix} \vec{k}^T & \vec{l}^T \end{pmatrix} J \begin{pmatrix} \vec{x} \\ \vec{p} \end{pmatrix} \right\} \widehat{\Mo}(\vec{k},\vec{l}) .
\end{multline*}
Since $\widehat{\Mo}(\vec{k},\vec{l}) \neq 0$ for all $\vec{k},\vec{l}$, by comparing the last two expressions we see that $\Mo^\prime = \Mo^{\prime\prime}$ if and only if
\[
\exp \left\{-\rmi\begin{pmatrix} \vec{k}^T & \vec{l}^T \end{pmatrix} J^\Mo \begin{pmatrix} \vec{x} \\ \vec{p} \end{pmatrix} \right\} = \exp \left\{-\rmi\begin{pmatrix} \vec{k}^T & \vec{l}^T \end{pmatrix} J \begin{pmatrix} \vec{x} \\ \vec{p} \end{pmatrix} \right\} , \qquad \forall \vec{x},\vec{p}\in\Rbb^n,\quad \forall\vec{k},\vec{l}\in\Rbb^m ,
\]
which in turn is equivalent to $J^\Mo = J$.
\end{proof}

\subsubsection*{Vector observables}
Let us point out the structure of the Gaussian approximate joint measurements of $\Qo$ and $\Po$.
\begin{proposition}
A bi-observable $\Mo^\sigma\in\Cscr$ is Gaussian if and only if the state $\sigma$ is Gaussian. In this case, the covariant bi-observable $\Mo^\sigma$ is Gaussian with parameters
$$
\mu^{\Mo^\sigma}=\mu^\sigma,\qquad V^{\Mo^\sigma} = V^\sigma,\qquad J^{\Mo^\sigma} = \id_{2n}.
$$
\end{proposition}
\begin{proof}
By comparing \eqref{eq:char_gauss}, \eqref{hatMosigma} and \eqref{eq:char_gauss_POVM}, and using the fact that $W(\vec{x}_1,\vec{p}_2) \propto W(\vec{x}_2,\vec{p}_2)$ if and only if $\vec{x}_1 = \vec{x}_2$ and $\vec{p}_1 = \vec{p}_2$, we have the first statement. Then, for $\sigma\in \Gscr$, we see immediately that $\Mo^\sigma$ is a Gaussian observable with the above parameters.
\end{proof}

We call $\Cscr^G$\marginpar{$\Cscr^G$} the class of the Gaussian covariant phase-space observables.
By \eqref{eq:hMo}, observing $\Mo^\sigma$ on a Gaussian state $\rho\in\Gscr$ yields the normal probability distribution
$\Mo^\sigma_\rho= \Ncal\left(\mu^\rho+\mu^\sigma ;\,V^\rho+V^\sigma\right)$,
with marginals
\begin{equation}\label{eq:hsigma}
\Mo^\sigma_{1,\rho} = \Ncal( \vec a^\rho+\vec a^\sigma  ;  A^\rho+A^\sigma), \qquad \Mo^\sigma_{2,\rho} = \Ncal( \vec b^\rho +\vec b^\sigma;  B^\rho+B^\sigma ).
\end{equation}
When $\vec a^\sigma = \vec 0$ and $\vec b^\sigma = \vec 0$, we have an \emph{unbiased measurement}.

\subsubsection*{Scalar observables}
We now study the Gaussian approximate joint measurements of the target observables $Q_{\vec{u}}$ and $P_{\vec{u}}$ defined in \eqref{QuPv}.
\begin{proposition}
A Gaussian bi-observable $\Mo$ with parameters $(\mu^\Mo, V^\Mo, J^\Mo)$ is in $\Cscr_{\vec u, \vec v}$ if and only if
$J^\Mo = J$, where $J$ is given by \eqref{eq:def_phi}. In this case, the condition \eqref{eq:def_pos_gauss_POVM} is equivalent to
\begin{equation}\label{eq:ineq_gaussian_R2}
V^\Mo_{11} \geq 0, \qquad V^\Mo_{22} \geq 0, \qquad V^\Mo_{11}V^\Mo_{22} \geq \frac{\hbar^2}{4} (\cos \alpha)^2 + (V^\Mo_{12})^2 .
\end{equation}
\end{proposition}
\begin{proof}
The first statement follows from Proposition \ref{prop:gauss_Tcov}. Then, the matrix inequality \eqref{eq:def_pos_gauss_POVM} reads
\[
V^\Mo  \pm \frac{\rmi\hbar}{2} \begin{pmatrix} 0 & \cos\alpha \\ - \cos\alpha & 0 \end{pmatrix}\geq 0,
\]
which is equivalent to \eqref{eq:ineq_gaussian_R2}.
\end{proof}

We write $\Cscr_{\vec u, \vec v}^G$\marginpar{$\Cscr_{\vec u, \vec v}^G$} for the class of the Gaussian $(\vec u,\vec v)$-covariant phase-space observables.
An observable $\Mo\in\Cscr^G_{\vec u, \vec v}$ is thus characterized by the couple $(\mu^\Mo,V^\Mo)$. From \eqref{eq:hMo} with $J^\Mo = J$ given by \eqref{eq:def_phi}, we get that
measuring $\Mo\in \Cscr_{\vec u, \vec v}^G$ on a Gaussian state $\rho$ yields the probability distribution
$
\Mo_\rho = \Ncal\left( \mu^\rho_{\vec u,\vec v}+ \mu^\Mo  ;\, V^\rho_{\vec u,\vec v} +V^\Mo \right)$.
Its marginals with respect to the first and second entry are, respectively,
\begin{equation}\label{Gapprjmeas}
\Mo_{1,\rho}  = \Ncal\left(\vec u\cdot\vec a^\rho+ a^\Mo;\,  \Var(\Qo_{\vec u,\rho})+V^\Mo_{11} \right), \qquad
\Mo_{2,\rho}  = \Ncal\left(\vec v\cdot\vec b^\rho+ b^\Mo;\,  \Var(\Po_{\vec v,\rho})+ V^\Mo_{22} \right) .
\end{equation}

\begin{example}\label{ex:Delta}
Let us construct an example of an approximate joint measurement of $Q_{\vec u}$ and $P_{\vec v}$, by using a noisy measurement of position along $\vec u$ followed by a sharp measurement of momentum along $\vec v$. Let $\Delta$ be a positive real number yielding the precision of the position measurement, and consider the POVM $\Mo$ on $\Rbb^2$ given by
\[
\Mo(A\times B) = \frac{1}{\sqrt{2\pi\Delta}} \int_A \exp\left\{-\frac{(x-Q_{\vec{u}})^2}{4\Delta}\right\} \Po_{\vec{v}}(B) \exp\left\{-\frac{(x-Q_{\vec{u}})^2}{4\Delta}\right\} \, \rmd x , \qquad \forall A,B\in\bor{\Rbb} .
\]
The characteristic function of $\Mo$ is
\begin{align*}
\widehat{\Mo}(k,l) & = \frac{1}{\sqrt{2\pi\Delta}} \int_{\Rbb} \rme^{\rmi kx} \exp\left\{-\frac{(x-Q_{\vec{u}})^2}{4\Delta}\right\} \left[\int_{\Rbb} \rme^{\rmi lp} \Po_{\vec{v}}(\rmd p)\right] \exp\left\{-\frac{(x-Q_{\vec{u}})^2}{4\Delta}\right\} \, \rmd x \\
& = \frac{1}{\sqrt{2\pi\Delta}} \int_{\Rbb} \exp\left\{\rmi kx -\frac{(x-Q_{\vec{u}})^2}{4\Delta}\right\} \rme^{\rmi lP_{\vec{v}}} \exp\left\{-\frac{(x-Q_{\vec{u}})^2}{4\Delta}\right\} \, \rmd x \\
& = \frac{\rme^{\rmi lP_{\vec{v}}}}{\sqrt{2\pi\Delta}} \int_{\Rbb} \exp\left\{\rmi kx-\frac{(x-Q_{\vec{u}}+\hbar l\vec{u}\cdot\vec{v})^2}{4\Delta}\right\} \exp\left\{-\frac{(x-Q_{\vec{u}})^2}{4\Delta}\right\} \, \rmd x \\
& = \frac{1}{\sqrt{2\pi\Delta}} \exp\left\{\rmi lP_{\vec{v}} - \frac{(\hbar l \cos\alpha)^2}{8\Delta}\right\} \int_{\Rbb} \exp\left\{\rmi kx-\frac{(x-Q_{\vec{u}}+\hbar l \cos\alpha/2)^2}{2\Delta}\right\} \, \rmd x \\
& = \exp\left\{\rmi lP_{\vec{v}} + \rmi k\left(Q_{\vec{u}}+\frac{\hbar l \cos \alpha}{2}\right)- \frac{\Delta}{2} k^2 - \frac{(\hbar \cos\alpha)^2}{8\Delta} l^2 \right\} \\
& = W(-\hbar l \vec{v} , \hbar k \vec{u}) \exp\left\{- \frac{\Delta}{2}\, k^2 - \frac{(\hbar \cos\alpha)^2}{8\Delta}\, l^2 \right\} .
\end{align*}
Therefore, $\Mo$ is a Gaussian bi-observable with parameters $a^\Mo =0$, $b^\Mo = 0$ and $J^\Mo =J$, where $J$ is given by \eqref{eq:def_phi} and $V^\Mo_{11}=\Delta $, $V^\Mo_{22}=\frac{(\hbar\cos\alpha)^2}{4\Delta}$ and $V^\Mo_{12}=0$.
This implies  $\Mo\in \Cscr_{\vec u, \vec v}^G$; in particular, the set $\Cscr_{\vec u, \vec v}^G$ is non-empty. Moreover, the lower bound $V^\Mo_{11}V^\Mo_{22} = \frac{\hbar^2}{4} (\cos \alpha)^2 $ is attained, cf.\ \eqref{eq:ineq_gaussian_R2}.
\end{example}

\begin{example}\label{ex:pvm}
Let us consider the case $\alpha=\pm \pi/2$; now the target observables $\Qo_{\vec u}$ and $\Po_{\vec v}$ are compatible and we can define a pvm $\Mo$ on $\Rbb^2$  by setting $\Mo(A\times B) = \Qo_{\vec u}(A)\Po_{\vec v}(B)$ for all $A,B\in\bor{\Rbb}$.   Its characteristic function is
\[
\widehat{\Mo}(k,l) = \int_{\Rbb} \rme^{i kx} \Qo_{\vec u}(\rmd x) \int_{\Rbb} \rme^{i lp} \Po_{\vec v}(\rmd p) = \rme^{i (k Q_{\vec u} + lP_{\vec v})} = W(-\hbar l \vec v, \hbar k \vec u) .
\]
Then, $\Mo\in \Cscr_{\vec u, \vec v}^G$ with parameters
$a^\Mo =0$, $ b^\Mo = 0$, $V^\Mo = 0$ and $J^\Mo = J$ given by \eqref{eq:def_phi}. Note that $\Mo$ can be regarded as the limit case of the observables of the previous example when $\cos \alpha = 0$ and $\Delta\downarrow 0$.
\end{example}

\section{Entropic MURs for position and momentum}\label{sec:entMURs}

In the case of two discrete target observables, in \cite{BarGT16} we found an entropic bound for the precision of their approximate joint measurements, which we named \emph{entropic incompatibility degree}. Its definition followed a three steps procedure. Firstly, we introduced an \emph{error function}: when the system is in a given state $\rho$, such a function quantifies the total amount of information that is lost by approximating the target observables by means of the marginals of a bi-observable; the error function is nothing else than the sum of the two relative entropies of the respective distributions. Then, we considered the worst possible case by maximizing the error function over $\rho$, thus obtaining an {\em entropic divergence} quantifying the approximation error in a state independent way. Finally, we got our index of the incompatibility of the two target observables by minimizing the entropic divergence over all bi-observables. In particular, when symmetries are present, we showed that the minimum is attained at some covariant bi-observables. So, the covariance followed as a byproduct of the optimization procedure, and was not a priori imposed upon the class of approximating bi-observables.

As we shall see, the extension of the previous procedure to position and momentum target observables is not straightforward, and peculiar problems of the continuous case arise. In order to overcome them, in this paper we shall fully analyse only a case in which explicit computations can be done: Gaussian preparations, and  Gaussian bi-observables, which we a priori assume to be covariant. We conjecture that the final result should be independent of these simplifications, as we shall discuss in Section \ref{sec:concl}.

As we said in Section \ref{jointmeas}, by ``approximate joint measurement'' we mean ``a bi-observable with the `right' covariance properties''.

\subsection{Scalar observables}\label{sec:uvbounds}

Given the directions $\vec u$ and $\vec v$, the target observables are $Q_{\vec u}$ and $P_{\vec v}$ in \eqref{QuPv} with pvm's $\Qo_{\vec u}$ and $\Po_{\vec v}$. For $\rho\in \Gscr$ with parameters $(\mu^\rho,V^\rho)$ given in \eqref{murho+V}, the target distributions $\Qo_{\vec u,\rho}$ and $\Po_{\vec v,\rho}$ are normal with means and variances \eqref{uvmom}.

An approximate joint measurements of $\Qo_{\vec u}$ and $\Po_{\vec v}$ is given by a covariant bi-observable $\Mo\in\Cscr_{\vec u, \vec v}$; then, we denote its marginals with respect to the first and second entry by $\Mo_1$ and $\Mo_2$, respectively. For a Gaussian covariant bi-observable $\Mo\in\Cscr^G_{\vec u, \vec v}$ with parameters $(\mu^\Mo,V^\Mo)$, the distribution of $\Mo$ in a Gaussian state $\rho$ is normal,
\[
\Mo_\rho=\Ncal\left(\mu^\rho + \mu^\Mo ;\, V^{\rho}_{\vec u,\vec v} + V^\Mo \right),
\]
so that its marginal distributions $\Mo_{1,\rho}$ and $\Mo_{2,\rho}$ are normal with means $\vec u\cdot\vec a^\rho+a^\Mo$ and $\vec v\cdot\vec b^\rho +  b^\Mo$ and variances
\begin{equation}\label{+var+}
\Var\left(\Mo_{1,\rho}\right)= \Var\left(\Qo_{\vec u,\rho}\right)+V^\Mo_{11},
\qquad
\Var\left(\Mo_{2,\rho}\right)= \Var\left(\Po_{\vec v,\rho}\right)+V^\Mo_{22}.
\end{equation}
Let us recall that $\abs{ \vec u}=1$, \ $\abs{\vec v}=1$, \ $\vec u\cdot \vec v=\cos \alpha$, and that by \eqref{RobUncert} and \eqref{eq:ineq_gaussian_R2}, we have
\begin{equation}\label{varvarcos}
\Var\left(\Qo_{\vec u,\rho}\right)\Var\left(\Po_{\vec v,\rho}\right)\geq \frac{\hbar^2}4 \left(\cos \alpha\right)^2,\qquad
V^\Mo_{11}\,V^\Mo_{22}\geq \frac{\hbar^2}4 \left(\cos \alpha\right)^2.
\end{equation}

\subsubsection{Error function}\label{sec:uverf}
The relative entropy is the amount of information that is lost when an approximating distribution is used in place of a target one. For this reason, we use it to give an informational quantification of the error made in approximating the distributions of sharp position and momentum by means of the marginals of a joint covariant observable.
\begin{definition}
Given the preparation $\rho\in \Sscr$ and the covariant bi-observable $\Mo\in \Cscr_{\vec u,\vec v}$, the \emph{error function} for the scalar case is the sum of the two relative entropies:
\begin{equation}\label{Salpha}
S(\rho,\Mo):=S(\Qo_{\vec u,\rho}\|\Mo_{1,\rho}) + S(\Po_{\vec v,\rho}\|\Mo_{2,\rho}).
\end{equation}
\end{definition}
The relative entropy is invariant under a change of the unit of measurement, so that the error function is scale invariant, too; indeed, it quantifies a relative error, not an absolute one.
In the Gaussian case the error function can be explicitly computed.

\begin{proposition}[Error function for the scalar Gaussian case] For $\rho\in \Gscr$ and $\Mo\in\Cscr_{\vec u,\vec v}^G$, the error function is
\begin{equation}\label{SalphaG}
S(\rho,\Mo)=\frac {\log \rme}2 \left[s(x)+s(y)+ \Delta(\rho,\Mo)\right],
\end{equation}
where
\begin{equation*}
x:=\frac{V^\Mo_{11}}{\Var\left(\Qo_{\vec u,\rho}\right)},\qquad y:=\frac{V^\Mo_{22}}{\Var\left(\Po_{\vec v,\rho}\right)},\qquad
\Delta(\rho,\Mo) :=\frac{(a^\Mo)^2}{\Var\left(\Mo_{1,\rho}\right)} + \frac{(b^\Mo)^2}{\Var\left(\Mo_{2,\rho}\right)},
\end{equation*}
and $s:[0,+\infty)\to[0,+\infty)$ is the following $\Ccal^\infty$ strictly increasing function with $s(0) = 0$:
\begin{equation}\label{def:s}
s(x):=\ln \left(1+x\right)-\frac x{1+x}.
\end{equation}
\end{proposition}
\begin{proof}
The statement follows by a straightforward combination of \eqref{Gsharpdistr}, \eqref{Grelentr}, \eqref{Gapprjmeas} and \eqref{Salpha}.
\end{proof}

Note that the error function does not depend on the mixed covariances $\vec{u} \cdot C^\rho \vec{v}$ and $V^\Mo_{12}$.
Note also that, if we select a possible approximation $\Mo$, then the error function $S(\rho,\Mo)$ decreases for states $\rho$ with increasing sharp variances $\Var\left(\Qo_{\vec u,\rho}\right)$ and $\Var\left(\Po_{\vec v,\rho}\right)$: the loss of information decreases when the sharp distributions make the approximation error negligible. Finally, note that
$$
s(x)+s(y)= \ln[(1+x)(1+y) ] +(1+x)^{-1} +(1+y)^{-1}-2,
$$
\[
1+ x
=\frac{\Var\left(\Mo_{1,\rho}\right)}{\Var\left(\Qo_{\vec u,\rho}\right)},\qquad 1+y=\frac{\Var\left(\Mo_{2,\rho}\right)}{\Var\left(\Po_{\vec v,\rho}\right)}.
\]
This means that, apart from the term $\Delta(\rho,\Mo)$ due to the bias, our error function $S(\rho,\Mo)$ only depends on the two ratios ``variance of the approximating distribution over variance of the target distribution''. Thus, in order to optimize the error function, one has to optimize these two ratios.
We use formula \eqref{SalphaG} to firstly give a state dependent MUR, and then, following the scheme of \cite{BarGT16}, a state independent MUR.

A lower bound for the error function can be found by minimizing it over all possible approximate joint measurements of $Q_{\vec u}$ and $P_{\vec v}$.
First of all, let us remark that this minimization makes sense because we consider only $(\vec{u},\vec{v})$-covariant bi-observables: if we minimized over all possible bi-observables, then the minimum would be trivially zero for every given preparation $\rho$. Indeed, the trivial bi-observable $\Mo(A\times B) = \Qo_{\vec{u},\rho}(A)\Po_{\vec{v},\rho}(B)\,\id$ yields $S(\rho,\Mo)=0$.

When minimizing the error function over all $(\vec{u},\vec{v})$-covariant bi-observables, both the minimum and the best measurement attaining it are state dependent. When $\alpha=\pm \pi/2$, the two target observables are compatible, so that their joint measurement trivially exists (see Example \ref{ex:pvm}) and we get $\inf_{\Mo\in \Cscr_{\vec u,\vec v}} S(\rho,\Mo)=0$. In order to have explicit results for any angle $\alpha$, we consider only the Gaussian case.

\begin{theorem}[State dependent MUR, scalar observables]\label{prop:min_sigma}
For every $\rho\in \Gscr$ and $\Mo\in\Cscr_{\vec u,\vec v}^G$,
\begin{equation}\label{lbalpharho}
S(\Qo_{{\vec u},\rho}\|\Mo_{1,\rho}) + S(\Po_{{\vec v},\rho}\|\Mo_{2,\rho})\geq c_\rho(\alpha),
\end{equation}
where the lower bound is
\begin{equation}\label{lbalpharhoB}
\begin{aligned}
& c_\rho(\alpha) =s\left(z_\rho\right)\log \rme\\
&\quad =\left(\log \rme\right) \left\{\ln \left(1+\frac{\hbar|\cos\alpha|}{2\sqrt{\Var\left(\Qo_{\vec u,\rho}\right) \Var\left(\Po_{\vec v,\rho}\right)}}\right)-\frac{\hbar|\cos\alpha|}{\hbar|\cos\alpha|+2\sqrt{\Var\left(\Qo_{\vec u,\rho}\right) \Var\left(\Po_{\vec v,\rho}\right)}}\right\},
\end{aligned}
\end{equation}
with
\begin{equation}\label{defzrho}
z_\rho:=\frac{\hbar\abs{\cos\alpha}}{2\sqrt{ \Var\left(\Qo_{\vec u,\rho}\right) \Var\left(\Po_{\vec v,\rho}\right)}} \in [0,1].
\end{equation}
The lower bound is tight and the optimal measurement is unique: $c_\rho(\alpha) = S(\rho,\Mo_*)$, for a unique $\Mo_*\in\Cscr_{\vec u,\vec v}^G$; such a  Gaussian $(\vec{u},\vec{v})$-covariant bi-observable is characterized by
\begin{equation}\label{sigma*}
\mu^{\Mo_*}=0, \quad V^{\Mo_*}_{12}=0,
\qquad V^{\Mo_*}_{11} =\frac \hbar 2 \,\sqrt{\frac{\Var\left(\Qo_{\vec u,\rho}\right)}{\Var\left(\Po_{\vec v,\rho}\right)}}\abs{\cos\alpha},
\quad V^{\Mo_*}_{22} =\frac \hbar 2 \,\sqrt{\frac{\Var\left(\Po_{\vec v,\rho}\right)}{\Var\left(\Qo_{\vec u,\rho}\right)}}\abs{\cos\alpha}.
\end{equation}
\end{theorem}
\begin{proof}
As already discussed, the case $\cos\alpha=0$ is trivial. If $\cos \alpha\neq 0$,
we have to minimize the error function \eqref{SalphaG} over $\Mo$. First of all we can eliminate the positive term $\Delta(\rho,\Mo)$ by taking an unbiased measurement. Then, since
$s$ is an increasing function, by the second condition in \eqref{varvarcos} we can also take $V^{\Mo_*}_{11}\,V^{\Mo_*}_{22}=
\frac{\hbar^2}4 \left(\cos \alpha\right)^2$. This implies $V^{\Mo_*}_{12}=0$ by \eqref{eq:ineq_gaussian_R2}. In this case the error function \eqref{SalphaG} reduces to
\[
S(\rho,\Mo_*)=\frac {\log \rme}2 \left(s(x)+s(z_\rho^{\,2}/x)\right),
\qquad  x =\frac{V^{\Mo_*}_{11}}{\Var\left(\Qo_{\vec u,\rho}\right)},
\]
with $z_\rho$ given by \eqref{defzrho};
by the first of \eqref{varvarcos}, we have $z_\rho\in (0,1]$.

Now, we can minimize the error function with respect to $x$ by studying its first derivative:
\[
\frac{\rmd \ }{\rmd x}\left(s(x)+s(z_\rho^{\,2}/x)\right)=\frac x{(1+x)^2}-\frac{z_\rho^{\,4}}{x(z_\rho^{\,2}+x)^2}
=\frac{\left(x^2-z_\rho^{\,2}\right)\left(x^2+2z_\rho^{\,2} x +z_\rho^{\,2}\right)}
{x\left(z_\rho^{\,2}+x\right)^2\left(1+x\right)^2}.
\]
Having $x>0$, we immediately get that $x=z_\rho$ gives the unique minimum. Thus
\[
S(\rho,\Mo)\geq S(\rho,\Mo_*)=s(z_\rho)\log \rme=\log(1+z_\rho)-\frac{z_\rho}{1+z_\rho}\,\log \rme,
\]
and
\[
V^{\Mo_*}_{11} =z_\rho \Var\left(\Qo_{\vec u,\rho}\right)\equiv\frac \hbar 2 \,\sqrt{\frac{\Var\left(\Qo_{\vec u,\rho}\right)}{\Var\left(\Po_{\vec v,\rho}\right)}}\abs{\cos\alpha},
\quad
V^{\Mo_*}_{22} =z_\rho \Var\left(\Po_{\vec v,\rho}\right)\equiv\frac \hbar 2 \,\sqrt{\frac{\Var\left(\Po_{\vec v,\rho}\right)}{\Var\left(\Qo_{\vec u,\rho}\right)}}\abs{\cos\alpha},
\]
which conclude the proof.
\end{proof}

\begin{remark}
The minimum information loss $c_\rho(\alpha)$ depends on both the preparation $\rho$ and the angle $\alpha$. When $\alpha\neq\pm\pi/2$, that is when the target observables are not compatible, $c_\rho(\alpha)$ is strictly grater than zero. This is a peculiar quantum effect: given $\rho$, $\vec u$ and $\vec v$, there is no Gaussian approximate joint measurement of $\Qo_{\vec u}$ and $\Po_{\vec v}$ that can approximate them arbitrarily well.
On the other side, in the limit $\alpha\to \pm\pi/2$, the lower bound $c_\rho(\alpha)$ goes to zero; so, the case of commuting target observables is approached with continuity.
\end{remark}

\begin{remark}
The lower bound $c_\rho(\alpha)$ goes to zero also in the classical limit $\hbar\to0$. This holds for every angle $\alpha$ and every Gaussian state $\rho$.
\end{remark}

\begin{remark}
Another case in which $c_\rho(\alpha)\to 0$ is the limit of large uncertainty states, that is, if we let the product $\Var\left(\Qo_{\vec u,\rho}\right)\,\Var\left(\Po_{\vec v,\rho}\right)\to\infty$: our entropic MUR disappears because, roughly speaking, the variance of (at least) one of the two target observables goes to infinity, its relative entropy vanishes by itself, and an optimal covariant bi-observable $\Mo_*$ has to take care of (at most) only the other target observable.
\end{remark}

\begin{remark}\label{scmacrolim}
Actually, something similar to the previous remark happens also at the macroscopic limit, and does not require the measuring instrument to be an optimal one; indeed, unbiasedness is enough in this case. This happens because the error function $S(\rho,\Mo)$ quantifies a relative error; even if the measurement approximation $\Mo$ is fixed, such an error can be reduced by suitably changing the preparation $\rho$. Indeed, if we consider the position and momentum of a macroscopic particle, for instance the center of mass of many particles, it is natural that its state has much larger position and momentum uncertainties than the intrinsic uncertainties of the measuring instrument; that is, $\frac{V^\Mo_{11}}{\Var\left(\Qo_{\vec u,\rho}\right)}\ll1$ and $\frac{V^\Mo_{22}}{\Var\left(\Po_{\vec v,\rho}\right)}\ll1$, implying that the error function \eqref{SalphaG} is negligible. In practice, this is a classical case: the preparation has large position and momentum uncertainties and the measuring instrument is relatively good. In this situation we do not see the difference between the joint measurement of position and momentum and their separate sharp observations.
\end{remark}

\begin{remark}
The optimal approximating joint measurement $\Mo_*\in\Cscr_{\vec u,\vec v}^G$ is unique; by \eqref{sigma*} it depends on the preparation $\rho$ one is considering, as well as on the directions $\vec u$ and $\vec v$. A realization of $\Mo_*$  is the measuring procedure of Example \ref{ex:Delta}.
\end{remark}

\begin{remark}
The MUR \eqref{lbalpharho} is scale invariant, as both the error function $S(\rho,\Mo)$ and the lower bound $c_\rho(\alpha)$ are such.
\end{remark}

\begin{remark}
For $\cos \alpha\neq 0$, we get $\inf_{\Mo\in \Cscr^G_{\vec u,\vec v}}S(\rho,\Mo)=s(z_\rho)\log\rme$, where $z_\rho$ is defined by \eqref{defzrho}. As $z_\rho$ ranges in the interval $(0,1]$, the quantity
$\inf_{\Mo\in \Cscr^G_{\vec u,\vec v}}S(\rho,\Mo)$ takes all the values in the interval $\left(0,\;1
-\frac{\log\rme}{2}\right]$,
so that
\begin{equation}\label{supinf3}
\sup_{\rho\in \Gscr}\inf_{\Mo\in \Cscr^G_{\vec u,\vec v}}S(\rho,\Mo)=1 -\frac{\log\rme}{2}.
\end{equation}
In order to get this result, we needed $\cos
\alpha\neq 0$; however, the final result does not depend on $\alpha$. Therefore, in the $\sup_{\rho}\inf_{\Mo}$-approach
of \eqref{supinf3},
the continuity from quantum to classical is lost.
\end{remark}

\subsubsection{Entropic divergence of $\Qo_{\vec u}, \Po_{\vec v}$ from $\Mo$}\label{sec:uventdiv}

Now we want to find an entropic quantification of the error made in observing $\Mo\in\Cscr_{\vec u,\vec v}$ as an approximation of $\Qo_{\vec u}$ and $\Po_{\vec v}$ in an arbitrary state $\rho$. The procedure of \cite{BarGT16}, already suggested in \cite[Sect.\ VI.C]{BLW14a} for a different error function, is to consider the worst case by maximizing the error function over all the states. However, in the continuous framework this is not possible for the error function \eqref{Salpha}; indeed, from \eqref{SalphaG} we get $\sup_{\rho\in \Gscr}S(\rho,\Mo)=+\infty$ even if we restrict to unbiased covariant bi-observables.

Anyway, the reason for $S(\rho,\Mo)$ to diverge is classical: it depends only on the continuous nature of $Q_{\vec u}$ and $P_{\vec v}$, without any relation to their (quantum) incompatibility. Indeed, as we noted in Section \ref{relentr}, if an instrument measuring a random variable $X\sim\Ncal(a;\alpha^2)$ adds an independent noise $\nu\sim\Ncal(b;\beta^2)$, thus producing an output $X+\nu\sim\Ncal(a+b;\alpha^2+\beta^2)$, then the relative entropy $S(X\|X+\nu)$ diverges for $\alpha^2\to0$; this is what happens if we fix the noise and we allow for arbitrarily peaked preparations.
Thus, the sum $S(\Qo_{\vec u,\rho}\|\Mo_{1,\rho}) + S(\Po_{\vec v,\rho}\|\Mo_{2,\rho})$ diverges if, fixed $\Mo$, we let
$\Var(\Qo_{\vec u,\rho})$ or $\Var(\Po_{\vec v,\rho})$ go to 0.

The difference between the classical and quantum frameworks emerges if we bound from below the variances of the sharp position and momentum observables.
Indeed, in the classical framework we have $\inf_{b,\beta^2}\sup_{\alpha^2\geq \vb}S(X\|X+\nu)=0$ for every $\vb>0$; the same holds for the sum of two relative entropies if no relation exists between the two noises.
On the contrary, in the quantum framework the entropic MURs appear due to the relation between the position and momentum errors occurring in any approximate joint measurement.

In order to avoid that $S(\rho,\Mo)\to +\infty$ due to merely classical effects, we thus introduce the following subset of the Gaussian states:\marginpar{$\Gscr_{\vec \vb}^{\vec u,\vec v}$}
\begin{equation}\label{eq:Guv_eps}
\Gscr_{\vec \vb}^{\vec u,\vec v}:=\left\{\rho \in \Gscr: \Var\left(\Qo_{\vec u,\rho}\right)\geq \vb_1, \  \Var\left(\Po_{\vec v,\rho}\right) \geq \vb_2\right\}, \qquad  \vb_i>0 ,
\end{equation}
and we evaluate the error made in approximating $\Qo_{\vec{u}}$ and $\Po_{\vec{v}}$ with the marginals of a $(\vec{u},\vec{v})$-covariant bi-observable by maximizing the error function over all these states.
\begin{definition}\label{def:div_scal}
The \emph{Gaussian $\vec\vb$-entropic divergence} of $\Qo_{\vec u}, \Po_{\vec v}$ from $\Mo\in\Cscr_{\vec u,\vec v}$ is
\begin{equation}\label{def:uvDG}
D^G_{\vec \vb}(\Qo_{\vec u}, \Po_{\vec v}\|\Mo):= \sup_{\rho\in \Gscr_{\vec \vb}^{\vec u,\vec v}} S(\rho,\Mo).
\end{equation}
\end{definition}

For Gaussian $\Mo$, depending on the choice of the thresholds $\vb_1$ and $\vb_2$, the divergence $D^G_{\vec \vb}(\Qo_{\vec u}, \Po_{\vec v}\|\Mo)$ can be easily computed or at least bounded.

\begin{theorem}\label{prop:uventdiv}
Let the bi-observable $\Mo\in\Cscr_{\vec u,\vec v}^G$ be fixed.
\begin{enumerate}[(i)]
\item For $\displaystyle\vb_1\vb_2\geq \frac{\hbar^2}4 \left(\cos \alpha\right)^2$, the divergence $D_{\vec\vb}^G(\Qo_{\vec u}, \Po_{\vec v}\|\Mo)$ is given by
\begin{equation}\label{DGuv}
D_{\vec\vb}^G(\Qo_{\vec u}, \Po_{\vec v}\|\Mo)=S(\rho_{\vec\vb}(\vec u,\vec v),\Mo)=\frac {\log \rme}2\left[s(x_{\vec\vb})+s(y_{\vec\vb})+\Delta(\vec\vb;\Mo)\right],
\end{equation}
where $\rho_{\vec\vb}(\vec u,\vec v)$ is any Gaussian state with $\Var\left(\Qo_{\vec u,\rho_{\vec\vb}(\vec u,\vec v)}\right)= \vb_1$ and  $\Var\left(\Po_{\vec v,\rho_{\vec\vb}(\vec u,\vec v)}\right) = \vb_2$, and
\[
x_{\vec\vb}:=\frac{V^\Mo_{11}}{\vb_1},  \quad y_{\vec\vb}:=\frac{V^\Mo_{22}}{\vb_2}, \qquad \Delta(\vec \vb;\sigma) :=\frac{(a^\Mo)^2}{V^\Mo_{11} +\vb_1} + \frac{(b^\Mo)^2}{V^\Mo_{22}+\vb_2}.
\]

\item For $\displaystyle\vb_1\vb_2<\frac{\hbar^2}4 \left(\cos \alpha\right)^2$, the divergence $D_{\vec\vb}^G(\Qo_{\vec u}, \Po_{\vec v}\|\Mo)$ is bounded from below by
\begin{equation}\label{DGuvB}
D_{\vec\vb}^G(\Qo_{\vec u}, \Po_{\vec v}\|\Mo)\geq S(\rho_{\vec\vb}(\vec u,\vec v),\Mo)=\frac {\log \rme}2\left[s(x_{\vec\vb})+s(y_{\vec\vb})+\Delta(\vec\vb;\Mo)\right],
\end{equation}
where $\rho_{\vec\vb}(\vec u,\vec v)$ is any Gaussian state with $\Var\left(\Qo_{\vec u,\rho_{\vec\vb}(\vec u,\vec v)}\right)= \vb_1$ and  $\displaystyle\Var\left(\Po_{\vec v,\rho_{\vec\vb}(\vec u,\vec v)}\right) = \frac{\hbar^2}{4\vb_1} \left(\cos \alpha\right)^2$, and
\[
x_{\vec\vb}:=\frac{V^\Mo_{11}}{\vb_1},  \quad y_{\vec\vb}:=\frac{4\vb_1\,V^\Mo_{22}}{\hbar^2\left(\cos \alpha\right)^2}, \qquad
\Delta(\vec \vb;\sigma) :=\frac{(a^\Mo)^2}{V^\Mo_{11} +\vb_1} + \frac{(b^\Mo)^2}{V^\Mo_{22}+\frac{\hbar^2}{4\vb_1} \left(\cos \alpha\right)^2}.
\]
\end{enumerate}
\end{theorem}
The existence of the above  states $\rho_{\vec\vb}(\vec u,\vec v)$  is guaranteed by Proposition \ref{admissvar}.
\begin{proof}
By Proposition \ref{admissvar}, maximizing the error function over the states in $\Gscr_{\vec \vb}^{\vec u,\vec v}$ is the same as maximizing \eqref{SalphaG} with \eqref{+var+} over the parameters $\Var\left(\Qo_{\vec u,\rho}\right)$ and  $\Var\left(\Po_{\vec v,\rho}\right)$ satisfying \eqref{varvarcos} and \eqref{eq:Guv_eps}.

(i) In the case $\displaystyle\vb_1\vb_2\geq \frac{\hbar^2}4 \left(\cos \alpha\right)^2$, the thresholds  themselves satisfy Heisenberg uncertainty relation, and so equality \eqref{DGuv} follows from the expression \eqref{SalphaG} and the fact the functions $s(x)$, $s(y)$, $\Delta(\rho,\Mo)$ are decreasing in $\Var\left(\Qo_{\vec u,\rho}\right)$ and  $\Var\left(\Po_{\vec v,\rho}\right)$.

(ii) In the case $\displaystyle\vb_1\vb_2<\frac{\hbar^2}4 \left(\cos \alpha\right)^2$, we have to take into account the relation \eqref{varvarcos} for $\Var\left(\Qo_{\vec u,\rho}\right)$ and  $\Var\left(\Po_{\vec v,\rho}\right)$: the supremum of $S(\rho,\Mo)$ is achieved when $\Var\left(\Qo_{\vec u,\rho}\right)\,\Var\left(\Po_{\vec v,\rho}\right)=\frac{\hbar^2}4 \left(\cos \alpha\right)^2$, with $\Var\left(\Qo_{\vec u,\rho}\right)\geq \vb_1$ and  $\Var\left(\Po_{\vec v,\rho}\right) \geq \vb_2$. Then inequality \eqref{DGuvB} follows by chosing $\Var\left(\Qo_{\vec u,\rho}\right)=\vb_1$ and $\displaystyle\Var\left(\Po_{\vec v,\rho}\right)=\frac{\hbar^2}{4\vb_1} \left(\cos \alpha\right)^2$.
\end{proof}
\begin{remark}
The conditions on the states $\rho_{\vec\vb}(\vec u,\vec v)$ do not depend on $\Mo$, but only on the parameters defining $\Gscr_{\vec\vb}^{\vec u,\vec v}$. Thus, in the case $\vb_1\vb_2\geq \frac{\hbar^2}4 \left(\cos \alpha\right)^2$, any choice of $\rho_{\vec\vb}(\vec u,\vec v)$ yields a state which is the worst one for every Gaussian approximate joint measurement $\Mo$.
\end{remark}

\subsubsection{Entropic incompatibility degree of $\Qo_{\vec u}$ and $\Po_{\vec v}$}\label{sec:uveid}

The last step is to optimize the state independent $\vec{\vb}$-entropic divergence \eqref{def:uvDG} over all the approximate joint measurements of $\Qo_{\vec u}$ and $\Po_{\vec v}$. This is done in the next definition.
\begin{definition}
The \emph{Gaussian $\vec\vb$-entropic incompatibility degree} of $\Qo_{\vec u}$, $\Po_{\vec v}$ is
\begin{equation}\label{def:cincuv}
c_{\rm inc}^G(\Qo_{\vec u},\Po_{\vec v};\vec\vb):=\inf_{\Mo\in\Cscr_{\vec u,\vec v}^G} D_{\vec\vb}^G(\Qo_{\vec u}, \Po_{\vec v}\|\Mo)\equiv \inf_{\Mo\in\Cscr_{\vec u,\vec v}^G} \sup_{\rho\in \Gscr_{\vec\vb}^{\vec u,\vec v}}  S(\rho,\Mo).
\end{equation}
\end{definition}

Again, depending on the choice of the thresholds $\vb_1$ and $\vb_2$, the entropic incompatibility degree  $c_{\rm inc}^G(\Qo_{\vec u},\Po_{\vec v};\vec\vb)$ can be easily computed or at least bounded.

\begin{theorem}\label{prop:uvincdeg}
\begin{enumerate}[(i)]
\item For $\displaystyle\vb_1\vb_2\geq \frac{\hbar^2}4 \left(\cos \alpha\right)^2$, the incompatibility degree $c_{\rm inc}^G(\Qo_{\vec u},\Po_{\vec v};\vec\vb)$ is given by
\begin{equation}\label{MUR0}
c_{\rm inc}^G(\Qo_{\vec u},\Po_{\vec v};\vec\vb) =(\log \rme)\left\{\ln\left(1+\frac {\hbar \abs{\cos \alpha}} {2\sqrt{\vb_1\vb_2}}\right)-\frac{\hbar \abs{\cos \alpha}}{2\sqrt{\vb_1\vb_2}+\hbar \abs{\cos \alpha}}\right\}.
\end{equation}
The infimum in \eqref{def:cincuv} is attained and the optimal measurement is unique, in the sense that
\begin{equation}\label{uvcinc1}
c_{\rm inc}^G(\Qo_{\vec u},\Po_{\vec v};\vec\vb)
= D_{\vec\vb}^G(\Qo_{\vec u}, \Po_{\vec v}\|\Mo_{\vec\vb})
\end{equation}
for a unique $\Mo_{\vec\vb}\in \Cscr^G_{\vec u,\vec v}$; such a bi-observable is characterized by
\begin{equation}\label{uvcinc2}
a^{\Mo_{\vec\vb}}=0,\quad b^{\Mo_{\vec\vb}}=0,\quad
V^{\Mo_{\vec\vb}}_{11} = \frac \hbar 2 \, \sqrt{\frac{\vb_1}{\vb_2}}\abs{\cos \alpha}, \quad V^{\Mo_{\vec\vb}}_{22} = \frac \hbar 2 \, \sqrt{\frac{\vb_2}{\vb_1}}\abs{\cos \alpha},\quad V^{\Mo_{\vec\vb}}_{12} =0.
\end{equation}
\item For $\displaystyle\vb_1\vb_2<\frac{\hbar^2}4 \left(\cos \alpha\right)^2$, the incompatibility degree $c_{\rm inc}^G(\Qo_{\vec u},\Po_{\vec v};\vec\vb)$ is bounded from below by
\begin{equation}\label{uvcinc1b}
c_{\rm inc}^G(\Qo_{\vec u},\Po_{\vec v};\vec\vb) \geq (\log \rme)\left\{\ln\left(2\right)-\frac{1}{2}\right\}.
\end{equation}
The latter bound is
\begin{equation}\label{uvcinc3}
(\log \rme)\left\{\ln\left(2\right)-\frac{1}{2}\right\}= S\big(\rho_{\vec\vb}(\vec u,\vec v),\Mo_{\vec\vb}\big) = \inf_{\Mo\in \Cscr^G_{\vec u,\vec v}} S\big(\rho_{\vec\vb}(\vec u,\vec v),\Mo\big),
\end{equation}
where the state \ $\rho_{\vec\vb}(\vec u,\vec v)$ \ is defined in item (ii) of Theorem \ref{prop:uventdiv} and \ $\Mo_{\vec\vb}$ \ is the bi-observable in $\Cscr^G_{\vec u,\vec v}$ such that
\begin{equation}\label{uvcinc4}
a^{\Mo_{\vec\vb}}=0,\qquad b^{\Mo_{\vec\vb}}=0,\qquad
V^{\Mo_{\vec\vb}}_{11} = \vb_1, \qquad V^{\Mo_{\vec\vb}}_{22} = \frac{\hbar^2}{4\vb_1} \left(\cos \alpha\right)^2,\qquad V^{\Mo_{\vec\vb}}_{12} =0.
\end{equation}
\end{enumerate}
\end{theorem}
\begin{proof}
(i) In the case $\displaystyle\vb_1\vb_2\geq \frac{\hbar^2}4 \left(\cos \alpha\right)^2$, due to \eqref{DGuv}, the proof is the same as that of Theorem \ref{prop:min_sigma} with the replacements $\Var\left(\Qo_{\vec u,\rho}\right)\mapsto \vb_1$ and $\Var\left(\Po_{\vec v,\rho}\right)\mapsto \vb_2$.

(ii) In the case $\displaystyle\vb_1\vb_2<\frac{\hbar^2}4 \left(\cos \alpha\right)^2$, starting from \eqref{DGuvB}, the proof is the same as that of Theorem \ref{prop:min_sigma} with the replacements $\Var\left(\Qo_{\vec u,\rho}\right)\mapsto \vb_1$ and $\Var\left(\Po_{\vec v,\rho}\right)\mapsto \frac{\hbar^2}{4\vb_1} \left(\cos \alpha\right)^2$.
\end{proof}

\begin{remark}[State independent MUR, scalar observables]\label{rem:MURscal}
By means of the above results, we can formulate a state independent entropic MUR for the position $Q_{\vec u}$ and the momentum $P_{\vec v}$ in the following way. Chosen two positive thresholds $\vb_1$ and $\vb_2$, there exists a preparation $\rho_{\vec\vb}(\vec u,\vec v)\in \Gscr_{\vec\vb}^{\vec u,\vec v}$ (introduced in Theorem \ref{prop:uventdiv}) such that, for all Gaussian approximate joint measurements $\Mo$ of $Q_{\vec u}$ and $P_{\vec v}$, we have
\begin{multline}\label{cincbounduv}
S(\Qo_{\vec u,\rho_{\vec\vb}(\vec u,\vec v)}\|\Mo_{1,\rho_{\vec\vb}(\vec u,\vec v)}) + S(\Po_{\vec v,\rho_{\vec\vb}(\vec u,\vec v)}\|\Mo_{2,\rho_{\vec\vb}(\vec u,\vec v)}) \\ {}\geq
\begin{cases}
\displaystyle(\log \rme)\left\{\ln\left(1+\frac {\hbar \abs{\cos \alpha}} {2\sqrt{\vb_1\vb_2}}\right)-\frac{\hbar \abs{\cos \alpha}}{2\sqrt{\vb_1\vb_2}+\hbar \abs{\cos \alpha}}\right\},&\text{if }\displaystyle\vb_1\vb_2\geq \frac{\hbar^2}4 \left(\cos \alpha\right)^2,\\
\displaystyle(\log \rme)\left\{\ln\left(2\right)-\frac{1}{2}\right\},&\text{if }\displaystyle\vb_1\vb_2<\frac{\hbar^2}4 \left(\cos \alpha\right)^2.
\end{cases}
\end{multline}
The inequality follows by \eqref{DGuv} and \eqref{MUR0} in the case $\vb_1\vb_2\geq \frac{\hbar^2}4 \left(\cos \alpha\right)^2$, and \eqref{uvcinc3} in the case $\vb_1\vb_2<\frac{\hbar^2}4 \left(\cos \alpha\right)^2$.

What is relevant is that, for every approximate joint measurement $\Mo$, the total information loss $S(\rho,\Mo)$ does exceed the lower bound \eqref{cincbounduv} even if the set of states $\Gscr^{\vec{u},\vec{v}}_{\vec \vb}$ forbids preparations $\rho$ with too peaked target distributions. Indeed, without the thresholds $\epsilon_1$, $\epsilon_2$, it would be trivial to exceed the lower bound \eqref{cincbounduv}, as we noted in Section \ref{sec:uventdiv}.

We also remark that, chosen $\vb_1$ and $\vb_2$, we found a single state $\rho_{\vec\vb}(\vec u,\vec v)$ in $\Gscr^{\vec{u},\vec{v}}_{\vec \vb}$ that satisfies \eqref{cincbounduv} for every $\Mo$, so that $\rho_{\vec\vb}(\vec u,\vec v)$ is a `bad' state for all Gaussian approximate joint measurements of position and momentum.

When $\vb_1\vb_2\geq \frac{\hbar^2}4 \left(\cos \alpha\right)^2$, the optimal approximate joint measurement $\Mo_{\vec \vb}$ is unique in the class of Gaussian $(\vec u,\vec v)$-covariant bi-observables; it depends only on the class of preparations $\Gscr_{\vec \vb}^{\vec u,\vec v}$: it is the best measurement for the worst choice of the preparation in the class $\Gscr_{\vec \vb}^{\vec u,\vec v}$.
\end{remark}

\begin{remark}
The entropic incompatibility degree $c_{\rm inc}^G(\Qo_{\vec u},\Po_{\vec v};\vec\vb)$ is strictly positive for $\cos \alpha \neq 0$ (incompatible target observables) and it goes to zero in the limits $\alpha\to \pm\pi/2$ (compatible observables), $\hbar\to0$ (classical limit), and $\vb_1\vb_2\to\infty$ (large uncertainty states).
\end{remark}

\begin{remark}
The scale invariance of the relative entropy extends to the error function $S(\rho,\Mo)$, hence to the divergence $D_{\vec\vb}^G(\Qo_{\vec u}, \Po_{\vec v}\|\Mo)$ and
the entropic incompatibility degree $c_{\rm inc}^G(\Qo_{\vec u},\Po_{\vec v};\vec\vb)$, as well as the entropic MUR \eqref{cincbounduv}.
\end{remark}

\subsection{Vector observables}

Now the target observables are $\vec Q$ and $\vec P$ given in \eqref{vecQP}, with pvm's $\Qo$ and $\Po$; the approximating bi-observables are the covariant phase-space observables $\Cscr$ of Definition \ref{def:cov_ph-sp}. Each bi-observable $\Mo\in\Cscr$ is of the form $\Mo = \Mo^\sigma$ for some $\sigma\in\Sscr$, where $\Mo^\sigma$ is given by \eqref{Mo}. $\Cscr^G$ is the subset of the Gaussian bi-observables in $\Cscr$, and $\Mo^\sigma\in\Cscr^G$ if and only if $\sigma$ is a Gaussian state.

We proceed to define the analogues of the scalar quantities introduced in Sections \ref{sec:uverf}, \ref{sec:uventdiv}, \ref{sec:uveid}. In order to do it, in the next proposition we recall some known results on matrices.
\begin{proposition}[\cite{Bha97,Carlen10,OhP93,Petz08}]\label{mon+matr}
Let $M_1$ and $M_2$ be $n\times n$ complex matrices such that $M_1> M_2> 0$. Then, we have  $0< M_1^{-1}< M_2^{-1}$. Moreover, if $s: \Rbb_+\to \Rbb$ is a strictly increasing continuous function, we have $\Tr \{s(M_1)\}> \Tr \{s(M_2)\}$.
\end{proposition}

\subsubsection{Error function}\label{sec:erf}

\begin{definition}
Given the preparation $\rho\in \Sscr$ and the covariant phase-space observable $\Mo^\sigma$, with $\sigma\in \Sscr$, the \emph{error function} for the vector case is the sum of the two relative entropies:
\begin{equation}
S(\rho,\Mo^\sigma):=S(\Qo_\rho\|\Mo_{1,\rho}^\sigma) + S(\Po_\rho\|\Mo_{2,\rho}^\sigma).
\end{equation}
\end{definition}
As in the scalar case, the error function is scale invariant, it quantifies a relative error, and we always have $S(\rho,\Mo^\sigma)>0$ because position and momentum are incompatible.
Indeed, since the marginals of a bi-observable $\Mo^\sigma\in \Cscr$ turn out to be convolutions of the respective sharp observables $\Qo$ and $\Po$ with some probability densities on $\Rbb^n$, $\Qo_{\rho} \neq \Mo_{1,\rho}^\sigma$ and $\Po_{\rho} \neq \Mo_{2,\rho}^\sigma$ for all states $\rho$; this is an easy consequence, for instance, of \cite[Problem 26.1, p.\ 362]{Bill86}.
In the Gaussian case the error function can be explicitly computed.

\begin{proposition}[Error function for the vector Gaussian case] For $\rho, \sigma\in \Gscr$, the error function has the two equivalent expressions:
\begin{subequations}
\begin{align}
S(\rho,\Mo^\sigma)&=\frac {\log \rme}2\left[\Tr \left\{s(E_{\rho,\sigma})+s(F_{\rho,\sigma})\right\}+ \vec a^\sigma \cdot (A^\rho+A^\sigma)^{-1}\vec a^\sigma + \vec b^\sigma \cdot (B^\rho+ B^\sigma)^{-1}\vec b^\sigma \right] \label{SEF}
\\ {} &=\frac {\log \rme}2\left[\Tr \left\{s(N_{\rho,\sigma}^{-1})+s(R_{\rho,\sigma}^{-1})\right]+ \vec a^\sigma \cdot (A^\rho+A^\sigma)^{-1}\vec a^\sigma + \vec b^\sigma \cdot (B^\rho+ B^\sigma)^{-1}\vec b^\sigma\right], \label{SNR}
\end{align}
\end{subequations}
where the function $s$ is defined in \eqref{def:s}, and
\begin{subequations}
\begin{equation}\label{def:EF}
E_{\rho,\sigma}:= (A^\rho)^{-1/2}A^\sigma (A^\rho)^{-1/2}, \qquad F_{\rho,\sigma}:= (B^\rho)^{-1/2}B^\sigma (B^\rho)^{-1/2},
\end{equation}
\begin{equation}\label{def:NR}
N_{\rho,\sigma}:= (A^\sigma)^{-1/2}A^\rho (A^\sigma )^{-1/2}, \qquad R_{\rho,\sigma}:= (B^\sigma)^{-1/2}B^\rho (B^\sigma )^{-1/2}.
\end{equation}
\end{subequations}
\end{proposition}

\begin{proof}
First of all, recall that
\begin{align*}
\Qo_\rho & = \Ncal(\vec a^\rho; A^\rho) , \qquad \Mo_{1,\rho}^\sigma = \Ncal(\vec a^\rho+\vec a^\sigma; A^\rho+A^\sigma) \\
\Po_\rho & = \Ncal(\vec b^\rho; B^\rho), \qquad \Mo_{2,\rho}^\sigma = \Ncal(\vec b^\rho+\vec b^\sigma; B^\rho+B^\sigma) .
\end{align*}
A direct application of \eqref{Grelentr} yields
\[
S(\Qo_\rho\|\Mo_{1,\rho}^\sigma) =\frac 12 \, \log \frac{\det(A^\rho+A^\sigma)} {\det A^\rho}+ \frac{\log \rme}2\left[\Tr\left\{(A^\rho+A^\sigma)^{-1} A^\rho-\openone\right\}+  \vec a^\sigma \cdot (A^\rho+A^\sigma)^{-1}\vec a^\sigma \right].
\]
We can transform this equation by using
\[
\frac{\det\left(A^\sigma +A^\rho\right)} {\det A^\rho }=\det\left[ (A^\rho)^{-1/2}\left(A^\sigma +A^\rho\right)(A^\rho)^{-1/2}
\right]=
\det\left(\openone +E_{\rho,\sigma}\right),
\]
\[
\ln\det\left(\openone +E_{\rho,\sigma}\right)=\Tr\left\{\ln\left(\openone +E_{\rho,\sigma}\right)\right\},
\]
\[
\Tr\left\{(A^\rho+A^\sigma)^{-1}A^\rho-\openone\right\}= \Tr\left\{(A^\rho)^{1/2}(A^\rho+A^\sigma)^{-1}(A^\rho)^{1/2}-\openone\right\} = -\Tr\left\{(\openone+E_{\rho,\sigma})^{-1}E_{\rho,\sigma}\right\}.
\]
This gives
\[
S(\Qo_\rho\|\Mo_{1,\rho}^\sigma) =\frac {\log \rme}2\left[\Tr\{s(E_{\rho,\sigma})\} + \vec a^\sigma \cdot (A^\rho+A^\sigma)^{-1}\vec a^\sigma \right].
\]
In the same way a similar expression is obtained for $S(\Po_\rho\|\Mo_{2,\rho}^\sigma)$ and \eqref{SEF} is proved.

On the other hand, by using
\[
\ln\frac{\det\left(A^\sigma +A^\rho\right)} {\det A^\rho }=\ln\frac{\det\left(\openone +N_{\rho,\sigma}\right)} {\det N_{\rho,\sigma} }=\ln\det\left( \openone +N_{\rho,\sigma}^{-1}\right)=\Tr\left\{ \ln\left(\openone +N_{\rho,\sigma}^{-1}\right)\right\},
\]
\[
\Tr\left\{(A^\rho+A^\sigma)^{-1}A^\rho-\openone\right\}=- \Tr\left\{\left( A^\rho+A^\sigma\right)^{-1}A^\sigma\right\} = -\Tr\left\{\left(\openone+N_{\rho,\sigma}^{-1}\right)^{-1}N_{\rho,\sigma}^{-1}\right\},
\]
and the analogous expressions involving $B^\rho$ and $R_{\rho,\sigma}$, one gets \eqref{SNR}.
\end{proof}

\subsubsection*{State dependent lower bound} In principle,
a state dependent lower bound for the error function could be found by analogy with Theorem \ref{prop:min_sigma}, by taking again the infimum over all joint covariant measurements, that is $\inf_\sigma S(\rho,\Mo^\sigma)$. By considering only Gaussian states $\rho$ and measurements $\Mo^\sigma$, from \eqref{ineqB>}, \eqref{SEF}, \eqref{def:EF} the infimum over $\sigma \in \Gscr $ can be reduced to an infimum over the matrices $A^\sigma$:
\[
\inf_{\sigma\in \Gscr} S(\rho,\Mo^\sigma)= \frac {\log \rme}2\,\inf_{A^\sigma}\Tr \left\{ s\left((A^\rho)^{-1/2}A^\sigma (A^\rho)^{-1/2}\right)+s\left(\frac{\hbar^2}4(B^\rho)^{-1/2} (A^\sigma)^{-1} (B^\rho)^{-1/2}\right)\right\}.
\]
The above equality follows since the monotonicity of $s$ (Proposition \ref{mon+matr}) implies that the trace term in \eqref{SEF} attains its minimum when $B^\sigma = \frac{\hbar^2}{4}(A^\rho)^{-1}$.
However, it remains an open problem to explicitly compute the infimum over the matrices $A^\sigma$ when the preparation $\rho$ is arbitrary.

Nevertheless, the computations can be done at least for a preparation $\rho_*$ of minimum uncertainty (Proposition \ref{prop:minun}). Indeed, by \eqref{mus2} we get
\[
\inf_{\sigma\in \Gscr} S(\rho_*,\Mo^\sigma)= \frac {\log \rme}2\,\inf_{A^\sigma}\Tr \left\{ s\left(E_{\rho,\sigma}\right)+s\left(E_{\rho,\sigma}^{\ -1}\right)\right\}.
\]
Now we can diagonalize $E_{\rho,\sigma}$ and minimize over its eigenvalues; since $s(x)+s(x^{-1})$ attains its minimum value at $x=1$, this procedure gives $E_{\rho,\sigma}=\openone$. So, by denoting by $\sigma_*$ the state giving the minimum, we have
\begin{equation}\label{quasicincsigma}
A^{\sigma_*} =A^{\rho_*}, \qquad B^{\sigma_*} =B^{\rho_*}=\frac{\hbar^2}4\left(A^{\rho_*}\right)^{-1},
\end{equation}
\begin{equation}\label{quasicinc}
\inf_{\sigma\in \Gscr} S(\rho_*,\Mo^\sigma)=S(\rho_*,\Mo^{\sigma_*})=ns(1)\log \rme.
\end{equation}

For an arbitrary $\rho\in\Gscr$, we can use the last formula to deduce an upper bound for $\inf_{\sigma\in \Gscr} S(\rho,\Mo^\sigma)$. Indeed, if $\rho_*$ is a minimum uncertainty state with $A^{\rho_*}=A^\rho$, then $B^\rho\geq\frac{\hbar^2}{4}(A^\rho)^{-1}=B^{\rho_*}$ by \eqref{weakineqB>}, and, using again the state $\sigma_*$ of \eqref{quasicincsigma}, we find
\[
\inf_{\sigma\in \Gscr} S(\rho,\Mo^\sigma)\leq S(\rho,\Mo^{\sigma_*})\leq S(\rho_*,\Mo^{\sigma_*})=ns(1)\log \rme.
\]
The second inequality in the last formula follows from \eqref{SNR}, \eqref{def:NR} and the monotonicity of $s$ (Prop. \ref{mon+matr}).

\subsubsection{Entropic divergence of $\Qo, \Po$ from $\Mo^\sigma$}\label{sec:entdiv}

In order to define a state independent measure of the error made in regarding the marginals of $\Mo^\sigma$ as approximations of $\Qo$ and $\Po$, we can proceed along the lines of the scalar case in Section \ref{sec:uventdiv}.
To this end, we introduce the following vector analogue of the Gaussian states defined in \eqref{eq:Guv_eps}:\marginpar{$\Gscr_{\vec \vb}$}
\begin{equation}\label{Gscreps}
\Gscr_{\vec \vb}:=\left\{\rho \in \Gscr: A^\rho\geq \vb_1\openone, \ B^\rho\geq \vb_2\openone\right\}, \qquad \vec \vb\equiv (\vb_1,\vb_2), \quad \vb_i>0.
\end{equation}
In the vector case, Definition \ref{def:div_scal} then reads as follows.
\begin{definition}
The \emph{Gaussian $\vec\vb$-entropic divergence} of $\Qo, \Po$ from $\Mo^\sigma\in\Cscr$ is
\begin{equation}\label{def:DG}
D^G_{\vec\vb}(\Qo, \Po\|\Mo^\sigma):= \sup_{\rho\in \Gscr_{\vec\vb}} S(\rho,\Mo^\sigma).
\end{equation}
\end{definition}
As in the scalar case, when $\Mo^\sigma$ is Gaussian, depending on the choice of the product $\vb_1\vb_2$, we can compute the divergence $D^G_{\vec\vb}(\Qo, \Po\|\Mo^\sigma)$ or at least bound it from below.
\begin{theorem}\label{prop:DG}
Let the bi-observable $\Mo^\sigma\in\Cscr^G$ be fixed.
\begin{enumerate}[(i)]
\item For $\displaystyle\vb_1\vb_2\geq\frac{\hbar^2}{4}$, the divergence $D^G_{\vec\vb}(\Qo, \Po\|\Mo^\sigma)$ is given by
\begin{multline}\label{DG}
D^G_{\vec\vb}(\Qo, \Po\|\Mo^\sigma)=  S(\rho_{\vec\vb},\Mo^\sigma)
=\frac {\log \rme}2\Bigl[\Tr \left\{ s\left( A^\sigma/\vb_1\right) + s\left( B^\sigma/\vb_2 \right)\right\} \\ {}+ \vec a^\sigma \cdot (A^\sigma +\vb_1\openone)^{-1}\vec a^\sigma + \vec b^\sigma \cdot (B^\sigma+\vb_2\openone)^{-1}\vec b^\sigma \Bigr],
\end{multline}
where $\rho_{\vec\vb}$ is any Gaussian state with $A^{\rho_{\vec\vb}}= \vb_1\openone$ and $B^{\rho_{\vec\vb}}= \vb_2\openone$.

\item
For $\displaystyle\vb_1\vb_2<\frac{\hbar^2}4$, the divergence $D^G_{\vec\vb}(\Qo, \Po\|\Mo^\sigma)$ is bounded from below by
\begin{multline}\label{DGB}
D^G_{\vec\vb}(\Qo, \Po\|\Mo^\sigma)\geq S(\rho_{\vec\vb},\Mo^\sigma)=\frac {\log \rme}2\Biggl[\Tr \left\{ s\left( A^\sigma/\vb_1\right)+s\left(4\vb_1B^\sigma/\hbar^2 \right)\right\} \\ {}+ \vec a^\sigma \cdot (A^\sigma +\vb_1\openone)^{-1}\vec a^\sigma + \vec b^\sigma \cdot \Bigg(B^\sigma+\frac{\hbar^2}{4\vb_1}\openone\Bigg)^{-1}\vec b^\sigma \Biggr],
\end{multline}
where $\rho_{\vec\vb}$ is any Gaussian state with $A^{\rho_{\vec\vb}}= \vb_1\openone$ and $\displaystyle B^{\rho_{\vec\vb}}= \frac{\hbar^2}{4\vb_1}\openone$.
\end{enumerate}
\end{theorem}
\begin{proof}
(i) In the case $\displaystyle\vb_1\vb_2\geq\frac{\hbar^2}{4}$, for $\rho\in \Gscr_{\vec\vb}$ we have $N_{\rho,\sigma}\geq \vb_1 (A^\sigma)^{-1}$ and $R_{\rho,\sigma}\geq \vb_2 (B^\sigma)^{-1}$; by Proposition \ref{mon+matr} we get
\[
\Tr\{ s(N_{\rho,\sigma}^{-1})\} \leq \Tr\left\{ s\left(A^\sigma/\vb_1\right)\right\},
\qquad  \Tr \{s(R_{\rho,\sigma}^{-1}) \}\leq \Tr \left\{ s\left(B^\sigma/\vb_2\right)\right\},
\]
\[
(A^\rho+A^\sigma)^{-1}\leq (\vb_1\openone+A^\sigma)^{-1}, \qquad (B^\rho+B^\sigma)^{-1}\leq (\vb_2\openone+B^\sigma)^{-1}.
\]
By using these inequalities in the expression \eqref{SNR}, we get \eqref{DG}.

(ii) In the case $\displaystyle\vb_1\vb_2<\frac{\hbar^2}{4}$, the lower bound \eqref{DGB} follows by evaluating $S(\rho,\Mo^\sigma)$ at the state $\rho=\rho_{\vec\vb}\in\Gscr_{\vec\vb}$ with $A^{\rho_{\vec\vb}}= \vb_1\openone$ and $\displaystyle B^{\rho_{\vec\vb}}= \frac{\hbar^2}{4\vb_1}\openone$.
\end{proof}

Note that $\rho_{\vec\vb}$ does not depend on $\sigma$, but only on the parameters defining $\Gscr_{\vec\vb}$:
again, in the case $\displaystyle\vb_1\vb_2\geq\frac{\hbar^2}{4}$, the error attains its maximum at a state which is independent of the approximate measurement.

\subsubsection{Entropic incompatibility degree of $\Qo$ and $\Po$}\label{sec:eid}

By analogy with Section \ref{sec:uveid}, we can optimize the $\vec{\vb}$-entropic divergence over all the approximate joint measurements of $\vec Q$ and $\vec P$.
\begin{definition}
The \emph{Gaussian $\vec\vb$-entropic incompatibility degree} of $\Qo$ and $\Po$ is
\begin{equation}\label{def:cG}
c_{\rm inc}^G(\Qo,\Po;\vec\vb):=\inf_{\sigma \in \Gscr}D_{\vec\vb}^G(\Qo, \Po\|\Mo^\sigma)\equiv \inf_{\sigma \in \Gscr} \sup_{\rho\in \Gscr_{\vec\vb}}  S(\rho,\Mo^\sigma).
\end{equation}
\end{definition}

Again, depending on the product $\vb_1\vb_2$, we can compute or at least bound $c_{\rm inc}^G(\Qo,\Po;\vec\vb)$ from below.

\begin{theorem}\label{prop:cinc}
\begin{enumerate}[(i)]
\item For $\displaystyle\vb_1\vb_2\geq \frac{\hbar^2}4$, the incompatibility degree $c_{\rm inc}^G(\Qo,\Po;\vec\vb)$ is given by
\begin{equation}\label{MUR1}
c_{\rm inc}^G(\Qo,\Po;\vec\vb)
= n\left(\log \rme\right) \left\{\ln\left(1+\frac {\hbar} {2\sqrt{\vb_1\vb_2}}\right)-\frac{\hbar}{2\sqrt{\vb_1\vb_2}+\hbar}\right\}.
\end{equation}
The infimum in \eqref{def:cG} is attained and the optimal measurement is unique, in the sense that
\begin{equation}\label{cinc1}
c_{\rm inc}^G(\Qo,\Po;\vec\vb)
= D_{\vec\vb}^G(\Qo, \Po\|\Mo^{\sigma_{\vec\vb}})
\end{equation}
for a unique $\sigma_{\vec\vb}\in \Gscr$; such a state is the minimal uncertainty state characterized by
\begin{equation}\label{cinc2}
\vec a^{\sigma_{\vec\vb}}=\vec 0, \qquad \vec b^{\sigma_{\vec\vb}}=\vec 0,\qquad A^{\sigma_{\vec\vb}} = \frac \hbar 2 \, \sqrt{\frac{\vb_1}{\vb_2}}\,\openone, \qquad B^{\sigma_{\vec\vb}} = \frac \hbar 2 \, \sqrt{\frac{\vb_2}{\vb_1}}\,\openone, \qquad C^{\sigma_{\vec\vb}} =0 .
\end{equation}

\item For $\displaystyle\vb_1\vb_2<\frac{\hbar^2}4 \left(\cos \alpha\right)^2$, the incompatibility degree $c_{\rm inc}^G(\Qo,\Po;\vec\vb)$ is bounded from below by
\begin{equation}\label{MUR2}
c_{\rm inc}^G(\Qo,\Po;\vec\vb) \geq n(\log \rme)\left\{\ln\left(2\right)-\frac{1}{2}\right\}.
\end{equation}
The latter bound is
\begin{equation}\label{cinc3}
n(\log \rme)\left\{\ln\left(2\right)-\frac{1}{2}\right\}= S(\rho_{\vec\vb},\Mo^{\sigma_{\vec\vb}}) = \inf_{\sigma\in \Gscr} S(\rho_{\vec\vb},\Mo^{\sigma}),
\end{equation}
where the preparation $\rho_{\vec\vb}$ is defined in item (ii) of Theorem \ref{prop:DG} and $\sigma_{\vec\vb}$ is the state in $\Gscr$ such that
\begin{equation}\label{cinc4}
\vec a^{\sigma_{\vec\vb}}=\vec 0, \qquad \vec b^{\sigma_{\vec\vb}}=\vec 0,\qquad A^{\sigma_{\vec\vb}} = \vb_1\,\openone, \qquad B^{\sigma_{\vec\vb}} = \frac{\hbar^2}{4\vb_1}\,\openone, \qquad C^{\sigma_{\vec\vb}} =0.
\end{equation}
\end{enumerate}
\end{theorem}

\begin{proof}
(i) In the case $\displaystyle\vb_1\vb_2\geq \frac{\hbar^2}4$, from the expression \eqref{DG} we get immediately \ $\vec a^{\sigma_{\vec\vb}}=\vec 0$, \ $\vec b^{\sigma_{\vec\vb}}=\vec 0$ \ and by \eqref{weakineqB>} we have \ $B^\sigma\geq \frac{\hbar^2}{4}(A^\sigma)^{-1}$. \ So, by \eqref{DG} and Propositions \ref{admissvar} and \ref{mon+matr}, we get $B^\sigma= \frac{\hbar^2}{4}(A^\sigma)^{-1}$, and
\[
\inf_{\sigma\in \Gscr}\sup_{\rho\in \Gscr_{\vec\vb}} S(\rho,\Mo^\sigma) =\frac {\log \rme}2\,\inf_{A^\sigma} \Tr \left\{ s\left( A^\sigma/\vb_1\right)+s\left(\frac{\hbar^2}{4\vb_2}(A^\sigma)^{-1}\right)\right\}.
\]
By minimizing over all the eigenvalues of $A^\sigma$, we get the minimum \eqref{MUR1}, which is attained if and only if $A^\sigma$ is as in \eqref{cinc2}. Hence, $A^{\sigma_{\vec \vb}}$ and $B^{\sigma_{\vec \vb}}$ are as in (87). This implies that any optimal state $\sigma_{\vec \vb}$ is a minimum uncertainty state; so, $C^{\sigma_{\vec \vb}}=0$ and the state $\sigma_{\vec \vb}$ is unique.

(ii) In the case $\displaystyle\vb_1\vb_2<\frac{\hbar^2}{4}$, by \eqref{weakineqB>} and Proposition \ref{mon+matr}, inequality \eqref{DGB} implies
\[
\inf_{\sigma\in \Gscr}\sup_{\rho\in \Gscr_{\vec\vb}} S(\rho,\Mo^\sigma) \geq \frac {\log \rme}2\,\inf_{A^\sigma} \Tr \left\{ s\left( A^\sigma/\vb_1\right)+s\left(\vb_1 (A^\sigma)^{-1}\right)\right\}.
\]
By minimizing over all the eigenvalues of $A^\sigma$, we get \eqref{MUR2}. Then \eqref{MUR2} holds for $\rho_{\vec\vb}$ as in item (ii) of Theorem \ref{prop:DG} and $\sigma_{\vec\vb}$ in \eqref{cinc4}.
\end{proof}

\begin{remark}[State independent MUR, vector observables]\label{rem:MURvect}
By means of the above results, we can formulate the following state independent entropic MUR for the position $\vec Q$ and momentum $\vec P$. Chosen two positive thresholds $\vb_1$ and $\vb_2$, there exists a preparation $\rho_{\vec\vb}\in \Gscr_{\vec\vb}$ (introduced in Theorem \ref{prop:DG}) such that, for all Gaussian approximate joint measurements $\Mo^\sigma$ of $\Qo$ and $\Po$, we have
\begin{multline}\label{cincbound}
S(\Qo_{\rho_{\vec\vb}}\|\Mo_{1,\rho_{\vec\vb}}^\sigma) + S(\Po_{\rho_{\vec\vb}}\|\Mo_{2,\rho_{\vec\vb}}^\sigma)\\
{}\geq
\begin{cases}
\displaystyle  n\left(\log \rme\right) \left\{\ln\left(1+\frac {\hbar} {2\sqrt{\vb_1\vb_2}}\right)-\frac{\hbar}{2\sqrt{\vb_1\vb_2}+\hbar}\right\},&\text{if }\displaystyle\vb_1\vb_2\geq \frac{\hbar^2}4,\\
\displaystyle n(\log \rme)\left\{\ln\left(2\right)-\frac{1}{2}\right\},&\text{if }\displaystyle\vb_1\vb_2<\frac{\hbar^2}4.
\end{cases}
\end{multline}
The inequality follows by \eqref{DG} and \eqref{MUR1} for $\vb_1\vb_2\geq \frac{\hbar^2}4$, and \eqref{cinc3} for $\vb_1\vb_2<\frac{\hbar^2}4$.

Thus, also in the vector case, for every approximate joint measurement $\Mo^\sigma$,
the total information loss $S(\rho,\Mo^\sigma)$ does exceed the lower bound \eqref{cincbound} even if $\Gscr_{\vec\vb}$ forbids preparations $\rho$ with too peaked target distributions. Moreover, chosen $\vb_1$ and $\vb_2$, one can fix again a single `bad' state $\rho_{\vec\vb}$ in $\Gscr_{\vec\vb}$ that satisfies \eqref{cincbound} for all Gaussian approximate joint measurements $\Mo^\sigma$ of $\vec Q$ and $\vec P$.

Whenever $\vb_1\vb_2\geq \frac{\hbar^2}4$, the optimal approximating joint measurement $\Mo^{\sigma_{\vec\vb}}$ is unique in the class of Gaussian covariant bi-observables; it corresponds to a minimum uncertainty state $\sigma_{\vec \vb}$ which
depends only on the chosen class of preparations $\Gscr_{\vec\vb}$, that is, on the thresholds $\epsilon_1$ and $\epsilon_2$: $\Mo^{\sigma_{\vec\vb}}$  is the best measurement for the worst choice of the preparation in that class.
\end{remark}

\begin{remark}
For $n=1$, the vector lower bound in \eqref{cincbound} reduces to the scalar lower bound found in \eqref{cincbounduv} for two parallel directions $\vec u$ and $\vec v$; for $n\geq 1$, the bound linearly grows with $n$.
\end{remark}

\begin{remark}
The entropic incompatibility degree $c_{\rm inc}^G(\Qo_{\vec u},\Po_{\vec v};\vec\vb)$ is strictly positive for $\cos \alpha \neq 0$ (incompatible target observables) and it goes to zero in the limit $\alpha\to \pm\pi/2$ (compatible observables), $\hbar\to0$ (classical limit), and $\vb_1\vb_2\to\infty$ (large uncertainty states).
\end{remark}

\begin{remark}
[The macroscopic limit.] Similarly to Remark \ref{scmacrolim} for scalar target observables, also the MUR \eqref{cincbound} is practically ineffective for macroscopic systems. Indeed, suppose we are concerned with position and momentum of a macroscopic particle, say the center of mass of a multi-particle system (in this case $n=3$). The states $\rho$ which can be practically prepared have macroscopic widths, say $\rho \in \Gscr_{\vec\vb}$ with `large' thresholds $\vec \vb$ and $\vb_1\vb_2\gg \hbar^2/4$. Then, we consider a measuring instrument $\Mo^{\sigma_*}$ having a high precision with respect to this class of states, but not necessarily reaching a precision near the quantum limits. For instance, let us take $\Mo^{\sigma_*}\in\Cscr^G$ with $A^{\sigma_*}= \delta_1\openone$, \ $B^{\sigma_*}= \delta_2\openone$, and $0<\delta_1\ll \vb_1 $, \ $0<\delta_2\ll \vb_2 $; we take it also unbiased: $\vec a^{\sigma_*}= \vec 0$, \ $b^{\sigma_*}= \vec 0$. Obviously, $\delta_1\delta_2\geq \hbar^2/4 $ must hold.
Then, $\forall \rho\in\Gscr_{\vec \vb}$ by \eqref{SEF} and \eqref{def:EF} we have
\[
E_{\rho,\sigma_*}=\frac{\delta_1}{A^\rho} \leq \frac {\delta_1}{\vb_1}\, \openone, \qquad F_{\rho,\sigma_*}=\frac{\delta_2 }{B^\rho} \leq \frac {\delta_2}{\vb_2}\, \openone,
\]
\[
0< S(\rho,\Mo^{\sigma_*})=\frac{\log \rme}2 \, \Tr\left\{ s(E_{\rho,\sigma_*}) + s(F_{\rho,\sigma_*})\right\} \leq \frac{n\log \rme}2 \left[ s(\delta_1/\vb_1) + s(\delta_2/\vb_2)\right].
\]
By \eqref{def:s} the function $s$ is increasing and in a neighborhood of zero it behaves as $s(x) \simeq x^2/2$; in the present case $\delta_1/\vb _1\ll 1$ and $\delta_2/\vb_2\ll 1$ and, so,  we have that the error function is negligible. This is practically a `classical' case: the preparation has `large' position and momentum uncertainties and the measuring instrument is `relatively good'. In this situation we do not see the difference between the joint measurement of position and momentum and their separate sharp distributions. Of course the bound \eqref{cincbound} continues to hold, but it is also negligible since $\epsilon_1 \epsilon_2 \gg \hbar^2/4$.
\end{remark}

\begin{remark}\label{rem:scinv}
The scale invariance of the relative entropy extends also in the vector case to the error function $S(\rho,\Mo^\sigma)$, the divergence $D_{\vec\vb}^G(\Qo, \Po\|\Mo^\sigma)$ and the entropic incompatibility degree $c_{\rm inc}^G(\Qo,\Po;\vec\vb)$, as well as the entropic MUR \eqref{cincbound}.
Indeed, let us consider the dimensionless versions of position and momentum \eqref{vecadim} and their associated projection valued measures $\widetilde \Qo$, $\widetilde \Po$ introduced in Section \ref{sec:PUR}. Accordingly, we rescale the joint measurement $\Mo^\sigma$ of \eqref{Mo} in the same way, obtaining the POVM
$
\widetilde \Mo^\sigma(B)  = \int_B \widetilde M^\sigma(\widetilde{\vec{x}},\widetilde{\vec{p}})\rmd \widetilde{\vec{x}} \rmd \widetilde{\vec{p}}$, 
\[
\widetilde M^\sigma(\widetilde{\vec{x}},\widetilde{\vec{p}})  =\frac{1}{\left(2\pi \lambda\right)^n}\, \exp\left\{\frac \rmi \lambda \left(\widetilde{\vec{p}}\cdot \widetilde{\vec Q} - \widetilde{\vec{x}}\cdot \widetilde{\vec P}\right)\right\} \Pi \sigma \Pi  \exp\left\{ - \frac \rmi \lambda \left(\widetilde{\vec{p}}\cdot \widetilde{\vec Q} - \widetilde{\vec{x}}\cdot \widetilde{\vec P}\right)\right\} .
\]
Here, both the vector variables $\widetilde{\vec{x}}$ and $\widetilde{\vec{p}}$, as well as the components of the Borel set $B$, are dimensionless. By the scale invariance of the relative entropy, the error function takes the same value as in the dimensioned case:
\begin{equation}\label{tildeerrf}
S(\widetilde\Qo_\rho\|\widetilde\Mo_{1,\rho}^\sigma) + S(\widetilde\Po_\rho\|\widetilde\Mo_{2,\rho}^\sigma)=
S(\Qo_\rho\|\Mo_{1,\rho}^\sigma) + S(\Po_\rho\|\Mo_{2,\rho}^\sigma).
\end{equation}
Then, the scale invariance holds for the entropic divergence and incompatibility degree, too:
\[
D_{\widetilde{\vec\vb}}^G(\widetilde\Qo, \widetilde\Po\|\widetilde\Mo^\sigma)=D_{\vec\vb}^G(\Qo, \Po\|\Mo^\sigma), \qquad c_{\rm inc}^G(\widetilde\Qo,\widetilde\Po;\widetilde{\vec\vb})=c_{\rm inc}^G(\Qo,\Po;\vec\vb),
\]
where $\displaystyle\widetilde \vb_1:= \frac{\varkappa \vb_1}{\hbar}$ and $\displaystyle\widetilde \vb_2:= \frac{\lambda^2 \vb_2 }{\varkappa \hbar}$. In particular $\displaystyle\ \widetilde\vb_1\widetilde\vb_2\geq\frac{\lambda^2}4 \ \iff \ \vb_1\vb_2\geq\frac{\hbar^2}4$ and, in this case, we have
$$
n\left(\log \rme\right) s\Big(\frac \lambda {2\sqrt{\widetilde \vb_1\widetilde \vb_2}}\Big)=c_{\rm inc}^G(\widetilde\Qo,\widetilde\Po;\widetilde{\vec\vb})=c_{\rm inc}^G(\Qo,\Po;\vec\vb)=n\left(\log \rme\right) s\Big(\frac \hbar {2\sqrt{\vb_1\vb_2}}\Big).
$$

\end{remark}

\section{Conclusions}\label{sec:concl}

We have extended the relative entropy formulation of MURs given in \cite{BarGT16} from the case of discrete incompatible observables  to a particular instance of continuous target observables, namely the position and momentum vectors, or two components of them along two possibly non parallel directions. The entropic MURs we found share the nice property of being scale invariant and well-behaved in the classical and macroscopic limits. Moreover, in the scalar case, when the angle spanned by the position and momentum components goes to $\pm\pi/2$, the entropic bound correctly reflects their increasing compatibility by approaching zero with continuity.

Although our results are limited to the case of Gaussian preparation states and covariant Gaussian approximate joint measurements, we conjecture that the bounds we found still hold for arbitrary states and general (not necessarily covariant or Gaussian) bi-observables. Let us see with some more detail how this should work in the case when the target observables are the vectors $\vec Q$ and $\vec P$.

The most general procedure should be to consider the error function
$S(\Qo_\rho\|\Mo_{1,\rho}) + S(\Po_\rho\|\Mo_{2,\rho})$ for an arbitrary POVM $\Mo$ on $\Rbb^n\times \Rbb^n$ and any state $\rho\in\Sscr$. First of all, we need states for which neither the position nor the momentum dispersion are too small; the obvious generalization of the test states \eqref{Gscreps} is
\begin{equation*}
\Sscr_{\vec\vb}:=\left\{\rho \in \Sscr_2: A^\rho\geq \vb_1\openone, \
B^\rho\geq \vb_2\openone\right\} ,  \qquad \vb_i>0 .
\end{equation*}
Then, the most general definitions of the entropic divergence and incompatibility degree are:
\begin{align}
D_{\vec\vb}(\Qo, \Po\|\Mo) & := \sup_{\rho\in \Sscr_{\vec\vb}}
\left[S(\Qo_\rho\|\Mo_{1,\rho}) + S(\Po_\rho\|\Mo_{2,\rho})\right], \label{eq:++} \\
c_{\rm inc}(\Qo,\Po;\vec\vb)& :=\inf_{\Mo}D_{\vec\vb}(\Qo, \Po\|\Mo).
\label{eq:+}
\end{align}
It may happen that $\Qo_\rho$ is not absolutely continuous with respect to $\Mo_{1,\rho}$, or $\Po_\rho$ with respect to $\Mo_{2,\rho}$; in this case, the error function and the entropic divergence take the value $+\infty$ by definition. So, we can restrict to bi-observables that are (weakly) absolutely continuous with respect to the Lebesgue measure. However, the true difficulty is that, even with this assumption, here we are not able to estimate \eqref{eq:++}, hence \eqref{eq:+}. It could be that the symmetrization techniques used in \cite{Wer04,BLW14a}  can be extended to the present setting, and one can reduce the evaluation of the entropic incompatibility index to optimizing over all covariant bi-observables. Indeed, in the present paper we a priori selected only covariant approximating measurements; we would like to understand if, among all approximating measurements, the relative entropy approach selects covariant bi-observables  by itself.
However, even if $\Mo$ is covariant, there remains the problem that we do not know how to evaluate \eqref{eq:++} if $\rho$ and $\Mo$ are not Gaussian. It is reasonable to expect that some continuity and convexity arguments should apply, and the bounds in Theorem \ref{prop:cinc} could be extended to the general case by taking dense convex combinations. Also the techniques used for the PURs in \cite{B-BM75,Beck75} could be of help in order to extend what we did with Gaussian states to arbitrary states. This leads us to conjecture:
\begin{equation}\label{eq:ultima}
c_{\rm inc}(\Qo,\Po;\vec\vb)=c_{\rm inc}^G(\Qo,\Po;\vec\vb) .
\end{equation}
Conjecture \eqref{eq:ultima} is also supported since the uniqueness of the optimal approximating bi-observable in Theorem \ref{prop:cinc}.(i) is reminiscent of what happens in the discrete case of two Fourier conjugated mutually unbiased bases (MUBs); indeed, in the latter case, the optimal bi-observable is actually unique among all the bi-observables, not only the covariant ones \cite[Theor.~5]{BarGT16}.
Similar considerations obviously apply also to the case of scalar target observables. We leave a more deep investigation of equality \eqref{eq:ultima} to future work.

As a final consideration, one could be interested in finding error/disturbance bounds involving sequential measurements of position and momentum, rather than considering all their possible approximate joint measurements. As sequential measurements are a proper subset of the set of all the bi-observables, optimizing only over them should lead to bounds that are greater than $c_{\rm inc}$. This is the reason for which in \cite{BarGT16} an error/disturbance entropic bound, denoted by $c_{\rm ed}$ and dinstinct from $c_{\rm inc}$, was introduced.
However, it was also proved that the equality $c_{\rm inc}=c_{\rm ed}$ holds when one of the target observables is discrete and sharp. Now, in the present paper, only sharp target observables are involved; although the argument of \cite{BarGT16} can not be extended to the continuous setting, the optimal approximating joint observables we found in Theorems \ref{prop:uvincdeg}.(i) and \ref{prop:cinc}.(i) {\em actually are} sequential measurements. Indeed, the optimal bi-observable in Theorem \ref{prop:uvincdeg}.(i) is one of the POVMs described in Examples \ref{ex:Delta} and \ref{ex:pvm} (see \eqref{uvcinc4}); all these bi-observables have a (trivial) sequential implementation in terms of an unsharp measurement of $\Qo_{\vec{u}}$ followed by sharp $\Po_{\vec{v}}$. On the other hand, in the vector case, it was shown in \cite[Corollary 1]{CHT11} that all covariant phase-space observables can be obtained as a sequential measurement of an unsharp version of the position $\Qo$ followed by the sharp measurement of the momentum $\Po$. Therefore, $c_{\rm inc}=c_{\rm ed}$ also for target position and momentum observables, in both the scalar and vector case.

\end{document}